\documentclass[a4paper,twocolumn,11pt,titlepage,unpublished]{quantumarticle}
\pdfoutput=1

\usepackage[T1]{fontenc}
\usepackage[utf8]{inputenc}
\usepackage[english]{babel}

\usepackage{amsmath,amssymb,amsthm}
\usepackage{mathtools}
\usepackage{bm}
\usepackage{fixmath}
\usepackage{physics}
\usepackage[square,comma,sort&compress,numbers]{natbib}

\usepackage{hyperref}

\usepackage{graphicx}
\usepackage{xcolor}
\usepackage{tikz}
\usetikzlibrary{arrows.meta,decorations.markings}
\usepackage{pgfplots}
\pgfplotsset{compat=1.18}
\usepackage{subcaption}
\usepackage{rotating}
\usepackage{placeins}

\usepackage{datetime2}
\usepackage{lipsum}
\usepackage{float}

\usepackage{algorithm}
\usepackage{algorithmic}

\newcounter{codimfootnote}

\newcommand{\transp}{\ensuremath{\mathrm{T}}}
\newcommand{\hc}{\ensuremath{\text{h.c.}}}
\newcommand{\C}{\mathbb{C}}
\newcommand{\R}{\mathbb{R}}
\newcommand{\cC}{\mathcal{C}}
\newcommand{\cS}{\mathcal{S}}
\newcommand{\da}{^\dagger}
\newcommand{\lam}{\lambda}
\newcommand{\blam}{{\bm{\lam}}}

\newcommand{\diag}{\operatorname{diag}}
\newcommand{\Real}{\operatorname{Re}}
\newcommand{\Img}{\operatorname{Im}}
\newcommand{\nxn}{{n \times n}}
\newcommand{\tnxtn}{{2n \times 2n}}
\newcommand{\mat}[1]{\begin{bmatrix}#1\end{bmatrix}}
\newcommand{\inv}{^{-1}}
\newcommand{\Sp}{\mathrm{Sp}^*(2n,\C)}

\newcommand{\jmvbmd}{joint-minimum-variation BMD}
\newcommand{\jmvabmd}{joint-minimum-variation ABMD}

\newtheorem{thm}{Theorem}
\newtheorem{property}{Property}
\theoremstyle{definition}
\newtheorem{definition}{Definition}
\newtheorem{remark}{Remark}

\begin{document}

\hypersetup{
  breaklinks=true,
  allcolors=quantumviolet
}

\title{Avoided crossings, degeneracies and Berry phases in the spectrum of quantum noise of driven-dissipative bosonic systems}

\author{Giuseppe Patera}
\email{giuseppe.patera@univ-lille.fr}
\affiliation{Univ. Lille, CNRS, UMR 8523 - PhLAM - Physique des Lasers Atomes et Molécules, Lille, 59000, France}
\orcid{0000-0003-3812-6568
}
\author{Alessandro Pugliese}
\email{alessandro.pugliese@uniba.it}
\orcid{0000-0003-3324-0169
}
\affiliation{Dipartimento di Matematica, Università degli Studi di Bari Aldo Moro, Via Orabona 4, Bari, 70125, Italy}
\maketitle

\begin{abstract}
Avoided crossings are fundamental phenomena in quantum mechanics and photonics that originate from the interaction between coupled energy levels and have been extensively studied in linear dispersive dynamics. Their manifestation in open, driven-dissipative systems, however, where nonlinear dynamics of quantum fluctuations come into play, remains largely unexplored. In this work, we analyze the hitherto unexplored occurrence of avoided and genuine crossings in the spectrum of quantum noise. We demonstrate that avoided crossings arise naturally when a single parameter is varied, leading to hypersensitivity of the associated singular vectors and suggesting the presence of genuine crossings (diabolical points) in nearby systems. We show that these spectral features can be deliberately designed, highlighting the possibility of programming the quantum noise response of photonic systems. As a notable example, such control can be exploited to generate broad, flat-band squeezing spectra -- a desirable feature for enhancing degaussification protocols. Our analysis is based on a detailed study of the Analytic Bloch-Messiah Decomposition (ABMD), which we use to characterize the parameter-dependent behavior of singular values and their corresponding vectors. This study provides new insights into the structure of multimode quantum correlations and offers a theoretical framework for the experimental exploitation of complex quantum optical systems.
\end{abstract}

\section{Introduction}
Avoided crossings are ubiquitous phenomena in physics, particularly in quantum mechanics and condensed matter, where they provide a hallmark of mechanisms such as quantum interference and energy-level repulsion~\cite{vonNeumann1929,Landau1932,Zener1932}. They have been extensively studied in contexts such as molecular spectroscopy~\cite{Herzberg1963,Yarkony1996}, quantum dots~\cite{Krenner2005,Stinaff2006}, superconductors~\cite{Simmonds2004,Wallraff2004}, and semiconductor nanostructures~\cite{Weisbuch1992,Winkler_book}. 
In photonics, avoided crossings emerge from the interaction between distinct optical modes, leading to the characteristic splitting and mutual repulsion of their resonant frequencies. These effects occur in a wide variety of optical systems—including waveguides~\cite{Jansen2011,Young2021}, photonic crystals~\cite{PalamaruJohnson2001}, microresonators~\cite{Herr2014}, and coupled-cavity architectures~\cite{Boriskina2007,Ryu2009}—and can be precisely engineered through the control of coupling strengths and system parameters. This control enables functionalities ranging from mode conversion and enhanced light–matter coupling to spectral filtering and dispersion engineering~\cite{Joannopoulos_book,Fan2011,Li2018}.

Such avoided crossings are features of the linear energy spectrum and arise from number-conserving quadratic interactions (e.g., beam-splitter-like couplings). In other words, they originate from mode-hopping terms associated with linear dispersive dynamics. In nonlinear photonic platforms—such as Kerr microresonators—these linear avoided crossings reshape the dispersion landscape and can be exploited to enable or tailor phenomena such as dispersive-wave emission and soliton formation~\cite{Xue2015,Yang2016,Yi2017}. However, in those cases the avoided crossing remains a property of the underlying linear mode structure, rather than a dynamical feature of the nonlinear quantum evolution itself.

It is far from obvious that such spectral features persist once the full driven–dissipative dynamics—including nonlinear or parametric processes—is taken into account in the presence of quantum noise. In this regime, after linearization around the classical steady state, quantum fluctuations are governed by quantum Langevin equations generated by an effective quadratic Hamiltonian that includes pair-generation terms. The evolution of quantum fluctuations is governed by smooth, frequency-dependent conjugate-symplectic transformations—which we term \textit{$\omega$-symplectic} transformations for brevity. The appropriate tool for their analysis is the Analytic Bloch–Messiah Decomposition (ABMD), which provides a smooth factorization of these transformations. This framework was introduced in~\cite{Gouzien2020}, developing a general approach where the system's transfer function is decomposed in terms of frequency-dependent supermodes. These modes—referred to as \textit{morphing supermodes} and physically interpreted as \textit{interferometers with memory effect} (IME)~\cite{Dioum2024}—have singular values that represent the corresponding frequency-dependent squeezing levels. The ABMD framework has been instrumental in revealing \textit{hidden squeezing}~\cite{Gouzien2023} and in identifying an optimal measurement strategy for its detection based on IMEs. Remarkably, a conceptually similar method, termed \textit{frequency-dependent principal component analysis}, has recently been proposed to characterize neuronal dynamics in brain activity~\cite{Calvo2024}.

\subsection{Contributions} 
In this work, we provide an in-depth analysis of the dynamics in driven–dissipative bosonic systems and show the occurrence of \emph{avoided crossings} in the spectrum of quantum noise, where the squeezing spectra approach each other and then abruptly veer apart.

Avoided crossings are typically accompanied by sharp variations in the corresponding singular vectors, a behaviour that makes them highly sensitive to small perturbations. In contrast, true intersections of singular values -- referred to as \emph{degeneracies} in this work, as \emph{diabolical points} in the physics literature, and as \emph{conical intersections} in the chemistry literature -- are generally not expected to occur unless the number of parameters is increased to match the \emph{codimension}\,\footnote{\label{fn:codim}Roughly speaking, the codimension of a phenomenon refers to the number of parameters that must be varied for the phenomenon to be generically observed. For a rigorous treatment of this concept, we refer the reader to \cite{Hirsch1976}. Throughout this work, dimension and codimension will always be considered with respect to the field of real numbers.}\setcounter{codimfootnote}{\value{footnote}} of the degeneracy.

A proper understanding of these features requires a formalism capable of tracking both singular values and their associated unitary factors in a smooth and physically consistent way. This role is naturally fulfilled by the Analytic Bloch–Messiah Decomposition (ABMD). We provide a detailed description of its fundamental properties and highlight its relevance in the analysis of mode dynamics.

Another contribution of this work is to clarify that the term ABMD actually encompasses a family of decompositions, all sharing the same singular values but generally differing in their unitary factors. We then focus our attention on a specific decomposition, which we name \emph{\jmvbmd} -- or \emph{ABMD}, depending on the smoothness of the matrix-valued function. We present an algorithm for computing it numerically (caution: smoothly computing a BMD/ABMD is not just collecting BMDs at separate parameter values, see Remark \ref{remark:caution}), and describe how it can be used to reveal avoided crossings and locate nearby degeneracies by exploiting certain topological properties of the unitary factors. All these concepts are illustrated through concrete examples of realistic physical systems in which such behaviors emerge, thereby demonstrating their relevance in experimental settings.

An interesting physical implication of these avoided crossings is the possibility of exploiting them to spectrally engineer the quantum-noise properties of the associated singular vectors. As pointed out by Asavanant \textit{et al.} in~\cite{Asavanant2017}, the quality of de-Gaussification of continuous-variable (CV) squeezed states depends on the flatness of the squeezing spectrum. Achieving this flatness has so far required narrowband filters, which, however, necessarily reduce the amount of nonclassical correlations available for the preparation of non-Gaussian states. We show that avoided crossings can be exploited to tailor quantum states featuring a broad, nearly flat-band squeezing spectrum.

\section{The playground: linear quantum Langevin equations and their decoupling}\label{sec:quantum}
We consider the nonlinear driven-dissipative evolution of a system of $n$ boson modes in the interaction picture. A standard linearization around a stable classical steady solution allows to reduce the description in terms of an effective quadratic Hamiltonian:
\begin{equation}
\hat{H} = \hbar \sum_{j,\ell}G_{j,\ell} \hat{a}_{j}^{\dagger}\hat{a}_{\ell}+ \frac{\hbar}{2}\sum_{j,\ell}[F_{j,\ell} \hat{a}_{j}^{\dagger}\hat{a}^{\dagger}_{\ell} + \hc],
\label{eq:hamiltonian}
\end{equation}
where $F$ and $G$ are $n \times n$ complex matrices with $F$ being symmetric ($F=F^{\transp}$), describing pair-production processes 
and $G$ being Hermitian ($G= G^{\dagger}$), describing mode-hopping processes.

In a driven-dissipative context, the Hamiltonian~\eqref{eq:hamiltonian} generates a set of coupled linear quantum Langevin equations that, in the quadrature representation, is given by
\begin{align}
\frac{\mathrm{d}\hat{\mathbf{R}}(t)}{\mathrm{d}t}&=
(-\Gamma+\mathcal{M})\hat{\mathbf{R}}(t)+
\sqrt{2\Gamma}\,\hat{\mathbf{R}}_{\mathrm{in}}(t)
\label{eq:quantum_lang_r}
\end{align}
where $\hat{\bm{R}}(t)=[\hat{x}_1(t),\ldots,\hat{x}_n(t),\hat{y}_1(t),\ldots,\hat{y}_n(t)]^\transp$ is the column vector of the amplitude and phase quadratures of the intracavity modes, $\hat{x}_j=(\hat{a}_j^\dag+\hat{a}_j)/\sqrt{2}$ and $\hat{y}_j=i(\hat{a}_j^\dag-\hat{a}_j)/\sqrt{2}$. The real $2n\times 2n$ diagonal matrix $\Gamma=\mathrm{diag}(\gamma_1,\ldots,\gamma_n,\gamma_1,\ldots,\gamma_n)$ accounts for the cavity damping rates that in general can be mode-dependent and that can result from multiple sources of losses or couplings with the external environment. The real $2n\times 2n$ mode interaction matrix $\mathcal{M}$
\begin{equation}
\mathcal{M}=
\left[
\begin{array}{cc}
\mathrm{Im}\left[G+F\right] & \mathrm{Re}\left[G-F\right]
\\
-\mathrm{Re}\left[G+F\right] & -\mathrm{Im}\left[G+F\right]^\transp
\end{array}
\right]
\label{eq:M_cali}
\end{equation}
encapsulates all mode interactions (mode-hopping and pair-production).

The quadratures $\hat{\mathbf{R}}_{\mathrm{out}}(t)$ of the system's output can be obtained through the input-output relations $\hat{\mathbf{R}}_{\mathrm{in}}(t)+\hat{\mathbf{R}}_{\mathrm{out}}(t)=\sqrt{2\Gamma}\,\hat{\mathbf{R}}(t)$~\cite{GardinerZoller}.
In the frequency domain, they are expressed in terms of the input quadratures and the transfer function $S(\omega)$
\begin{equation}
\hat{\mathbf{R}}_{\mathrm{out}}(\omega)= S(\omega)\hat{\mathbf{R}}_{\mathrm{in}}(\omega), \label{eq:Rout}
\end{equation}
where $S(\omega)$ is a complex matrix-valued function of the continuous parameter $\omega$ given by the expression
\begin{equation}
S(\omega)=\sqrt{2\Gamma}\left(i\omega I_{2n}+\Gamma-\mathcal{M}\right)^{-1}\sqrt{2\Gamma}-I_{2n}
\label{S}
\end{equation}
with $I_{2n}$ the $2n\times 2n$ identity matrix. The input spectral quadrature operators satisfy the commutation rule $\left[\hat{\bm{R}}_{\mathrm{in}}(\omega),\hat{\bm{R}}_{\mathrm{in}}^{\transp}(\omega')\right]=i\Omega\delta(\omega+\omega')$, 
where $\Omega=\begin{bmatrix} 0 & I_n \\  -I_n & 0\end{bmatrix}$ is the $n$-mode symplectic form and $I_n$ is the $n\times n$ identity matrix. The output quadratures $\hat{\bm{R}}_{\mathrm{out}}(\omega)$ are the Fourier transform of \textit{bona fide} boson quadrature operators in time domain, since $S(\omega)$ has the two following properties~\cite{Gouzien2020}:
\begin{property}
$S$ is conjugate symmetric: $S(-\omega)=S^*(\omega)$, for all $\omega$. This assures the reality of $S$ in time domain.
\end{property}
\begin{property}\label{property:omega-symp}
$S$ is ``$\omega$-symplectic'' which stands for a smooth (possibly analytic) matrix-valued function that is conjugate symplectic for all $\omega$, see Definition \ref{def:conj_simpl}. This follows from the fact that $\mathcal{M}$ is an Hamiltonian matrix, i.e. ${(\Omega\,\mathcal{M})}^{\transp}=\Omega\,\mathcal{M}$, and $\Gamma$ is a skew-Hamiltonian, i.e. ${(\Omega\,\Gamma)}^{\transp}=-\Omega\,\Gamma$.
\end{property}
\noindent
We recall the following definition:
\begin{definition}\label{def:conj_simpl}
A matrix $A\in\C^\tnxtn$ is said to be \emph{conjugate symplectic} if it satisfies $A\Omega A\da=\Omega$.
The conjugate symplectic group is noted as $\Sp$.
\end{definition}
\noindent
We also introduce the following definition throughout the paper:
\begin{definition}
    A matrix-valued function of $\bm{\lambda}=(\lambda_1,\ldots,\lambda_p)\in\mathbb{R}^{p}$ is said to be ``$\bm{\lambda}$-symplectic'' if it is smooth (possibly analytic) and conjugate symplectic for all values of $\bm{\lambda}$.
\end{definition}
While eq.~\eqref{eq:Rout} formally expresses the solution of~\eqref{eq:quantum_lang_r}, its highly coupled form makes it difficult to exploit it in order to characterize the system's dynamics and the corresponding quantum properties. A better approach involves a set of normal modes (a.k.a. supermodes) that decouple the system's dissipative dynamics and map the CV multimode entangled state into a collection of statistically uncorrelated (anti-)squeezed states~\cite{Patera2010,Roslund2014}. 

Brute force diagonalization does not always lead to supermodes that correspond to physical observables. In first instance, because $\mathcal{M}$ may be non diagonalizable; in second instance, because one needs $\mathcal{M}$ to be diagonalizable via an orthogonal and symplectic matrix $O$.
\begin{align}
\mathcal{M}=O \Lambda O^\transp.
\end{align}
This ``symplectic diagonalization" exists only for $\mathcal{M}$ symmetric~\cite{delaCruz2016,Fassbender2005,TheseElie}, which is equivalent to having $G=0$.

As demonstrated in Ref.~\cite{Gouzien2020}, in the general case of a quadratic Hamiltonian, the modes that decouple~\eqref{S} are necessarily ``morphing''. 
Morphing supermodes are a generalization of standard (static) supermodes and are expressed as linear combinations of the initial system modes with frequency-dependent coefficients that vary with $\omega$. They are obtained by performing the analytic Bloch-Messiah decomposition (ABMD) of the transfer function
\begin{align}
S(\omega)=U(\omega)D(\omega)V^\dagger(\omega).
\label{ABMD}
\end{align}
In this expression, $U(\omega)$ and $V(\omega)$ are unitary and $\omega$-symplectic functions that characterize the supermodes structure.
In particular, the input and output quadratures of morphing supermodes are respectively given by
\begin{align}
\hat{\bm{R}}^{(s)}_{\mathrm{in}}(\omega) &= V^\dagger(\omega)\hat{\bm{R}}_{\mathrm{in}}(\omega),
\label{morsup in}
\\
\hat{\bm{R}}^{(s)}_{\mathrm{out}}(\omega) &= U^\dagger(\omega)\hat{\bm{R}}_{\mathrm{out}}(\omega).
\label{morsup out}
\end{align}
We note that, since the vacuum state is unitary invariant, the input supermodes can be disregarded when the input is assumed to be in this state. In this case, assuming an arrangement of the analytic singular values 
such that $D(\omega)= \diag(d_1(\omega),\ldots,d_n(\omega)|\,d_1^{-1}(\omega),\ldots,d_n^{-1}(\omega))$, with $d_j(\omega)\ge1$ for $j=1,\ldots,n$ and for all $\omega$, the quadratures for $j=1,\ldots,n$ are amplified
\begin{align}
\hat{R}^{(s)}_{\mathrm{out},j}(\omega)&=
d_j(\omega)\hat{R}_{\mathrm{in},j}(\omega)
\end{align}
and the output quadratures for $j=n+1,\ldots,2n$ are squeezed
\begin{align}
\hat{R}^{(s)}_{\mathrm{out},j}(\omega)&=
d_j^{-1}(\omega)\hat{R}_{\mathrm{in},j}(\omega).
\end{align}
The corresponding quantum state, since it is produced by the action of a quadratic Hamiltonian on an input Gaussian state (the vacuum), will also be a Gaussian state whose quantum properties are fully captured, up to a displacement in the phase-space, by the spectral covariance matrix
\begin{align}
\sigma_{\mathrm{out}}(\omega)&=
\frac{1}{2}\left\langle 
\hat{\mathbf{R}}_{\mathrm{out}}(\omega)
\hat{\mathbf{R}}_{\mathrm{out}}^{\transp}(-\omega)
+
\left(\hat{\mathbf{R}}_{\mathrm{out}}(-\omega)
\hat{\mathbf{R}}_{\mathrm{out}}^{\transp}(\omega)\right)^{\transp}\right\rangle.
\label{eq:covar_spect}   
\end{align}
In the supermodes basis, $\sigma_{\mathrm{out}}^{(s)}(\omega)=\mathrm{diag}(d_1^2(\omega),\ldots,d_n^{2}(\omega)|d_1^{-2}(\omega),\ldots,d_n^{-2}(\omega))$ is diagonal with
anti-squeezed entries for $j=1,\ldots,n$ and squeezed entries for $j=n+1,\ldots,2n$. Therefore the frequency-dependent singular values of the system's transfer function $S(\omega)$ determine the squeezing and anti-squeezing spectra of their corresponding supermodes.

\section{Smooth decompositions}\label{sec:LinAlg}

\subsection{Static SVD}
It is well known (e.g., see \cite{Horn_Johnson_2012}) that any matrix $A\in\C^{m\times n}$ admits a Singular Value Decomposition (SVD): 
\begin{equation*}
    A=U\Sigma V\da,
\end{equation*}
where
\begin{itemize}
    \item[-] $\Sigma=\diag(\sigma_1,\sigma_2,\ldots,\sigma_q)\in\R^{m\times n}$, $q=\min\{m,n\}$, $\sigma_1\ge\sigma_2\ge\ldots\ge\sigma_q\ge 0$;
    \item[-] $U=\mat{u_1, u_2, \ldots, u_m}\in\C^{m\times m}$ and $V=\mat{v_1, v_2, \ldots, v_n}\in\C^{n\times n}$ are unitary matrices (partitioned by columns).
\end{itemize}

The values $\sigma_j$ are called the singular values of $A$, and the corresponding columns of $U$ and $V$ are called its left and right singular vectors, respectively. 
If $\sigma_j=\sigma_{j+1}$ for some $j$, we say that $A$ is \emph{degenerate}. In this work, the word \emph{degeneracy} will always refer to situations where some singular values coincide, and our attention is restricted to square matrices ($m = n$).

Our focus is on the SVD of parameter-dependent matrices, in particular on typical phenomena that concern their singular values and singular vectors. In this perspective, degeneracy of singular values and its impact on the non-uniqueness of the SVD factors play a critical role, and we elaborate more on this below. 

The degree of non-uniqueness of the SVD depends on the number of distinct singular values. If the singular values of $A\in\C^\nxn$ are all distinct, its singular vectors are unique (only) up to an arbitrary phase factor. That is, if $A=U\Sigma V\da$ is an SVD of $A$ with $\sigma_1>\sigma_2>\ldots>\sigma_n$, then any other SVD must be of the form $A={\widetilde U} \Sigma {\widetilde V}\da$, where $\widetilde U=U\Theta$ and $\widetilde V=V\Theta$, with $\Theta=\diag\big(e^{i\theta_1},e^{i\theta_2},\ldots,e^{i\theta_n}\big)$ for some $\theta_j\in\R$, $j=1,\ldots,n$. However, if $A\in\C^\nxn$ has fewer than $n$ distinct singular values, then the unitary factors of its SVD gain additional degrees of freedom. For instance, suppose that $A=U\Sigma V\da$ is an SVD of $A$ with $\sigma_1=\sigma_2>\sigma_3>\ldots>\sigma_n$. Then, $A={\widetilde U} \Sigma {\widetilde V}\da$ is an SVD of $A$ if and only if $\widetilde U=U\mat{W & 0 \\ 0 & \Theta}$ and $\widetilde V=V\mat{W & 0 \\ 0 & \Theta}$, where $W$ is a $2 \times 2$ complex unitary matrix and $\Theta=\diag\big(e^{i\theta_3},e^{i\theta_4},\ldots,e^{i\theta_n}\big)$ with $\theta_j\in\R$, $j=3,\ldots,n$. If the number of distinct singular values drops below $n-1$, the freedom in choosing the unitary factors of the SVD increases further.
\smallskip

\subsection{Dynamic SVD}
The degeneracy of singular values becomes even more significant in the study of the SVD of parameter-dependent matrices. Consider a complex matrix-valued function 
\begin{equation*}
    A:\blam=(\lam_1,\lam_2,\ldots,\lam_p)\in\R^p\mapsto A(\blam)\in\C^\nxn\,,
\end{equation*}
which has $k\ge1$ continuous derivatives (or is possibly analytic) with respect to $\blam$. Throughout this work, we refer to such matrix functions as \emph{smooth}. If $A(\blam)$ has distinct non-zero singular values for all $\blam$, then it is known (see \cite{hsieh_sibuya_1971}) that it has an SVD
\begin{equation*}
    A(\blam)=U(\blam)\Sigma(\blam) V(\blam)\da,
\end{equation*}
where the factors are as smooth as $A$, regardless of the number of parameters $p$. The situation changes drastically if $A(\blam)$ is degenerate for some value of $\blam$. Typically, near a point $\blam \in \R^p$, $p \geq 2$, where $A(\blam)$ is degenerate, its ordered singular values are merely continuous functions of $\blam$, while the singular vectors corresponding to degenerate singular values will lose continuity (e.g., see \cite{kato1980}). 
For one-parameter matrix functions ($p=1$), the behavior differs depending on whether the function is analytic or $\cC^k$. An analytic complex matrix-valued function of one parameter $\lam\in\R$ always has analytic SVD factors, even when singular values are degenerate (again, see \cite{kato1980}). For $\cC^k$ matrix functions, $k\ge 1$, the singular vectors can lose smoothness (or even continuity) at points of degeneracy, while the singular values can still retain their smoothness. We refer to \cite{Dieci1999} for a comprehensive study of the smoothness of matrix factorizations that depend on one (real) parameter.

\subsection{Codimension of degeneracies}
The analysis above highlights that the smoothness of the factors of the SVD is intimately related to the degeneracy of the singular values. This motivates the following question: how likely is it for a matrix $A\in\C^\nxn$ to have degenerate singular values? The most appropriate way to answer this question is to look at the dimension of the set of degenerate matrices:
\begin{multline*}
    \cS_n=\{A\in\C^\nxn: A \text{ has fewer than } n \text{ distinct} \\ \text{singular values}\}.
\end{multline*}
It is known that $\cS_n$, as a subset of $\C^\nxn$, has real dimension $2n^2-3$ (see \cite{Dieci1999}), where $2n^2$ is just the real dimension of $\C^\nxn$. Hence, the real codimension of $\cS_n$ is $3$ \footnotemark[\value{codimfootnote}]. To see why, it is insightful to look at the $2\times 2$ case, where $\cS_2$ is just the set of positive scalar multiples of $2\times 2$ unitary matrices:
\begin{equation*}
    \cS_2=\{\sigma W, \text{ with } \sigma\ge0 \text{ and } W\in\C^{2 \times 2} \text{ unitary }\},
\end{equation*}
that has dimension 5, which is 3 less than the dimension of $\C^{2 \times 2}$.

Generically, a one-parameter or two-parameters family of $\nxn$ complex matrices will not intersect $\cS_n$. This has the following important consequences:
\begin{enumerate}
    \item A generic smooth complex matrix-valued function depending on \textbf{one or two} parameters has distinct singular values, and hence has a
    smooth singular value decomposition (see \cite{Dieci1999});
    \item A generic smooth complex matrix-valued function depending on \textbf{three} parameters is expected to be degenerate at isolated points in parameters' space.
\end{enumerate}

The codimension of the degenerate set $\cS_n$ may decrease due to additional symmetries in the matrices. For example, the codimension is 2 for real symmetric matrices (see \cite{DiPu_realSVD}) and complex symmetric matrices (see \cite{DiPaPu_takagi}). Another instance of reduced codimension will be encountered later, in Section \ref{sec:phys_ex}.

\smallskip

\subsection{Static and dynamic Bloch-Messiah decomposition}
From this point onward, we focus on conjugate symplectic matrices (see Definition \eqref{def:conj_simpl}). For consistency with the notation introduced in Section \ref{sec:quantum}, we denote the singular values of a $\tnxtn$ conjugate symplectic matrix by $d_1,\ldots,d_n, d_1\inv,\ldots, d_n\inv$, where $d_1\ge\ldots\ge d_n\ge 1$, and denote by $D$ the diagonal matrix of its singular values:
\begin{equation}\label{eq:D_expression}
    D=\mat{D_1 & 0 \\ 0 & D_1\inv}=\diag(D_1,D_1\inv),
\end{equation}
where $D_1=\diag(d_1,\ldots,d_n)$.

First, we observe that every conjugate symplectic matrix admits an SVD in which the unitary factors preserve the symplectic structure.

\begin{thm}[\cite{XU20031}]\label{thm:sympSVD_Xu} Let $S\in\C^\tnxtn$ be conjugate symplectic. Then, there exist $U, V\in\C^\tnxtn$ unitary conjugate symplectic and $D_1=\diag(d_1,\ldots,d_n)$ with $d_1\ge\ldots
\ge d_n\ge1$ such that 
 \begin{equation}\label{eq:sympSVD_Xu}
     S=UDV^\dag,
 \end{equation}
 where $D$ is as in \eqref{eq:D_expression}.
\end{thm}

The factorization in \eqref{eq:sympSVD_Xu} is known as Bloch-Messiah decomposition (BDM). 
We recall that a conjugate symplectic unitary matrix $W\in\C^\tnxtn$ must have the form
\begin{equation}\label{eq:conjSympUnitaryForm}
    W=\mat{W_{11} & W_{12} \\ -W_{12} & W_{11}},
\end{equation}
where each block is $\nxn$, e.g. see \cite[p. 14]{paige1981schur}; it is also easy to see that a conjugate symplectic real diagonal matrix $D\in\R^\tnxtn$ must have the form $D=\diag\left(D_{1},D_1\inv\right)$.

In general, given a conjugate symplectic matrix, standard software for the computation of the SVD does not return a BDM. Below we describe a simple algorithm to obtain the BMD of a conjugate symplectic matrix $S$ from \emph{any} given SVD of $S$. We specialize it to the case where $S$ has distinct singular values, as this is both generic and sufficient for our scope.

\begin{algorithm}[H]
\algsetup{linenodelimiter=.}
\caption{BMD from any SVD}
\label{algo:Bloch-Messiah}
\begin{algorithmic}[1]
\REQUIRE $S\in\C^\tnxtn$ conjugate symplectic with distinct singular values.
\smallskip
\ENSURE $U, V\in\C^\tnxtn$ and $D_1\in\R^\nxn$ as in Theorem \ref{thm:sympSVD_Xu}.
\medskip
\STATE Let $S=U\Sigma V^\dag$ be an SVD of $S$ with $U,V\in \C^\tnxtn$ unitary, and singular values arranged in strictly decreasing order along the diagonal of $\Sigma\in\R^\tnxtn$.
\STATE Let $\Pi$ be the $\tnxtn$ permutation matrix for which $\Pi^T\Sigma\Pi=\mat{D_1 & 0 \\ 0 & D_1\inv}$ with $(D_1)_{jj}=(\Sigma)_{jj}$, $j=1,\ldots,n$, and update the unitary factors as follows:
\begin{equation*}
U\leftarrow U\Pi,\quad V\leftarrow V\Pi.
\end{equation*}
\STATE Partition $U$ in $\nxn$ blocks, $U=\mat{U_{11} & U_{12} \\ U_{21} & U_{22}}$, and compute real numbers $\theta_1,\ldots,\theta_n$ such that $\mat{U_{11} \\ -U_{21}}=\mat{U_{22} \\ U_{12}}\diag\big(e^{i\theta_1},\ldots,e^{i\theta_n}\big)$.
\STATE Set $\Theta=\diag\big(e^{-i\theta_1/2},\ldots,e^{-i\theta_n/2},e^{i\theta_1/2},\ldots,e^{i\theta_n/2}\big)$ and update the unitary factors as follows: 
\begin{equation*}
    U\leftarrow U\Theta,\quad V\leftarrow V\Theta.
\end{equation*}
\end{algorithmic}
\end{algorithm}

The existence of the values $\theta_j \in \R$, $j = 1, \ldots, n$, needed in step 3 of the algorithm, follows directly from Theorem \ref{thm:sympSVD_Xu} and from the degree of uniqueness of the BMD for a matrix with distinct singular values. In practice, these values are computed by taking the (principal) logarithm of the diagonal entries of the matrix  
$$
\begin{bmatrix} U_{22}^\dag & U_{12}^\dag \end{bmatrix} 
\begin{bmatrix} U_{11} \\ -U_{21} \end{bmatrix}.
$$

It is important to highlight that, if the input matrix $S$ is a smooth function of a real parameter, Algorithm \ref{algo:Bloch-Messiah} can be arranged so as to produce a smooth BMD. This result holds because each step of the algorithm can be performed in a way that preserves smoothness. In particular, step 1 must yield a smooth SVD, and, in step 3, special care is required: to retain smoothness, switching branches of the logarithm may be necessary when computing the $\theta_j$’s.

We formalize this observation in the following theorem.

\begin{thm}\label{thm:smoothBlochMessiahDec} Any smooth $\omega$-symplectic matrix function $S(\omega)$, $\omega\in\R$, having distinct singular values for all $\omega$, admits a smooth Bloch-Messiah decomposition
\begin{equation}\label{eq:smoothBMD}
    S(\omega)=U(\omega)D(\omega)V^\dag(\omega)
\end{equation}
where all the factors are $\omega$-symplectic.
\end{thm}
If $S$ is analytic, then the decomposition of Theorem~\ref{thm:smoothBlochMessiahDec} can be called ABMD as in ~\cite{Gouzien2020}.

It is worth noting that other choices are possible in step 4 of Algorithm \ref{algo:Bloch-Messiah}. In fact, as discussed earlier, the decomposition in \eqref{eq:smoothBMD} is inherently non-unique. Given a smooth Bloch-Messiah decomposition as in \eqref{eq:smoothBMD}, any decomposition
$$
    S(\omega)=\widetilde{U}(\omega)D(\omega)\widetilde{V}^\dag(\omega)
$$
is also a smooth Bloch-Messiah decomposition of $S(\omega)$ as long as we take
$$
\widetilde{U}(\omega) = U(\omega)\Theta(\omega), \quad \widetilde{V}(\omega) = V(\omega)\Theta(\omega),
$$
where $\Theta(\omega)$ is a diagonal \emph{phase matrix}:
\begin{equation}\label{eq:phase_matrix}
\Theta(\omega) = \diag\left(e^{i\theta_1(\omega)}, \ldots, e^{i\theta_n(\omega)}, e^{i\theta_1(\omega)}, \ldots, e^{i\theta_n(\omega)}\right)
\end{equation}
with each $\theta_j(\omega)$ being an arbitrary smooth real-valued function of $\omega\in\R$.

Numerically computing a smooth BDM in practice requires resolving its inherent non-uniqueness. Notably, the issue of selecting an appropriate smooth decomposition for a one-parameter matrix-valued function is not specific to our situation. For instance, it also arises in the eigendecomposition of Hermitian matrix-valued functions with distinct eigenvalues, where it is standard practice to adopt the so-called \emph{Minimum Variation Decomposition} (MVD), originally introduced in \cite{BunseGerstner1991}. This is closely linked to the definition of the \emph{Berry phase} around closed loops, as introduced in \cite{berry1984quantal} and further discussed in \cite{DIPU_LAA2012}.  

In this work, a key consideration in resolving the non-uniqueness of the decomposition is to do so in a way that is conducive to the detection of parameter values where singular values become degenerate. With this goal in mind, we introduce the following definition, where we adopt the approach recently proposed in \cite{DIPU_LAA2024}.

\begin{definition} Let $S(\omega)$, $\omega\in\R$, be a $\omega$-symplectic matrix function (see Property \ref{property:omega-symp}) having distinct singular values for all $\omega$, and consider a smooth Bloch-Messiah decomposition of $S$ as in Theorem \ref{thm:smoothBlochMessiahDec}. We will call \eqref{eq:smoothBMD} a \emph{\jmvbmd} over the interval $[0,\tau]$ if the pair $(U,V)$ minimizes the quantity
\begin{equation}\label{eq:joint_MVD_integral}
\int_0^\tau\sqrt{\norm{\dot U(\omega)}^2_{\mathrm{F}}+\norm{\dot V(\omega)}^2_{\mathrm{F}}}\,d\omega\,,
\end{equation}
where $\norm{\cdot}_\mathrm{F}$ denotes the Frobenius norm. If $S(\omega)$ is analytic in $\omega$, we refer to \eqref{eq:smoothBMD} as \emph{\jmvabmd}.
\end{definition}

\begin{remark}
As we have already seen, the inherent non-uniqueness of the unitary factors in the BMD allows for (infinitely) many different ways of constructing a smooth BMD for a given $\omega$-symplectic matrix-valued function. We remark that the recipe used by Gouzien \textit{et al.} (see supplemental notes in ~\cite{Gouzien2020}, page 5), which amounts to imposing the condition $\diag(V(\omega)^\dag\dot V(\omega))=0$, corresponds~\footnote{This follows from Theorem 2.6 in \cite{DIPU_LAA2012}.} to minimizing only the variation of the left unitary factor $V$. Consequently, the resulting ABMD differs from the \jmvabmd, in which the variations of both unitary factors contribute to the minimized quantity. Both recipes are equally valid if one is solely interested in the spectrum of the singular values. However, as we will illustrate in Section~\ref{sec:phys_ex}, the \jmvbmd\ is the mandatory choice for detecting degeneracies, especially when their codimension may be lower than anticipated.
\end{remark}

If $S(\omega)$ is further periodic, then the unitary factors of its \jmvbmd\ acquire a phase over one period. This phase is, in fact, the natural extension of the Berry phase to the smooth BMD, as discussed in \cite{DIPU_LAA2024}. Appropriately monitoring this phase allows us to infer whether a pair of singular values has become degenerate within a given region in parameter space. We formalize this fact in the following theorem, where, to simplify the discussion, we restrict ourselves to spherical regions in parameter space. For further details and extensions, we refer the reader to \cite{DIPU_LAA2012}.

\begin{thm}[Adapted from \cite{DIPU_LAA2012}]\label{thm:codim3detect}
    Consider a $\bm{\lambda}$-symplectic function depending on 3 parameters:
\begin{equation*}
    S:\blam=(\lam_1,\lam_2,\lam_3)\in\R^3\mapsto S(\blam)\in\C^\tnxtn\,,
\end{equation*}
and let $\mathbb{S}\subset\R^3$ be a sphere parametrized in spherical coordinates by:
\begin{equation*}
    \left\{\begin{array}{l} 
    \lambda_1(\varphi,\theta)=c_1+\rho\cos(\varphi)\sin(\theta) \\
    \lambda_2(\varphi,\theta)=c_2+\rho\sin(\varphi)\sin(\theta) \\
    \lambda_3(\varphi,\theta)=c_3+\rho\cos(\theta)
    \end{array} \right.,
\end{equation*}
$\varphi\in[0, 2\pi]$, $\theta\in[0,\pi]$. For every $\theta\in[0,\pi]$ and every $j=1,\ldots,2n$, consider the phase $\alpha_j(\theta)$ accrued by the $j$-th singular vector of the \jmvbmd\ of the matrix-valued function
\begin{equation*}
    \varphi\in[0,2\pi]\mapsto S(\blam(\varphi,\theta)))
\end{equation*}
over the interval $[0,2\pi]$. Note that each $\alpha_j$ can be chosen to be a continuous function of $\theta$ (see \cite{DIPU_LAA2012}). If $\alpha_j(\pi)\ne\alpha_j(0)$ for some $j$, then $S$ must be degenerate at some point $\blam_0$ inside the sphere $\mathbb{S}$, and the degeneracy must involve the singular value $\sigma_j$.
\end{thm}

As we have already pointed out, 
special symmetries in $S$ may reduce the number of parameters that we should expect to vary in order to observe a degeneracy.

In Section \ref{sec:phys_ex}, Theorem \ref{thm:codim3detect} will be used as a topological tool to detect parameter values where certain symplectic matrix functions are degenerate. In order to do so, we will need to compute the \jmvbmd\ of some matrix functions. Below, we outline the algorithm we have used to accomplish this task. The algorithm is an adaptation of the one proposed in \cite{DIPAPU_SIMAX2013} to our context.

\begin{algorithm}[H]
\algsetup{linenodelimiter=.}
\caption{Computation of a Minimum Variation Bloch-Messiah Decomposition}
\label{algo:smooth_BMD}
\begin{algorithmic}[1]
\REQUIRE An $\omega$-symplectic matrix-valued function $$S:\omega\in [0,\tau]\mapsto S(\omega)\in\C^\tnxtn$$ having distinct singular values for all $\omega\in[0,\tau]$.
\smallskip
\ENSURE Finite sequences
\begin{equation*}
\begin{array}{l}
    0=\omega_0<\omega_1<\ldots<\omega_N=\tau, \\
    U_k\approx U(\omega_k)\in\C^\tnxtn, \\
    V_k\approx V(\omega_k)\in\C^\tnxtn, \\
    D_k\approx D(\omega_k)\in\R^\nxn,
\end{array}
\end{equation*}
$k=0, 1, \ldots, N$, where 
\begin{equation*}
S(\omega)=U(\omega)D(\omega) V^\dag(\omega)
\end{equation*}
is a \jmvbmd\ of $S(\omega)$.
\medskip
\STATE Compute a BDM of $S$ at $\omega_0=0$: $$S(\omega_0)=U_0 D_0 V_0^\dag;$$
\STATE For $k=1, 2, \ldots, N$:

\begin{itemize}
    \item[2.1.] Appropriately choose $h>0$, set $\omega_k=\omega_{k-1}+h$, and compute a BDM of $S$ at $\omega_k$:
$$S(\omega_k)=\widehat{U}_k D_k \widehat{V}_k^\dag.$$
    \item[2.2.] Compute a phase matrix
    \begin{equation*}
    \Theta = \diag(e^{i\theta_1}, \ldots, e^{i\theta_n}, e^{i\theta_1}, \ldots, e^{i\theta_n}),
    \end{equation*}
    such that
    \begin{equation}\label{eq:procrustes}
    \sqrt{
    \norm{U_{k-1}-\widehat{U}_k \Theta}^2_\mathrm{F} +
    \norm{V_{k-1}-\widehat{V}_k \Theta}^2_\mathrm{F}
    }
    \end{equation}
    is minimized over all possible choices of $\theta_1,\ldots,\theta_n\in\R$;  
    \item[2.3.] Set
    $$
    U_k=\widehat{U}_k \Theta,\ V_k=\widehat{V}_k \Theta.
    $$
\end{itemize}
\end{algorithmic}
\end{algorithm}

Algorithms~\ref{algo:Bloch-Messiah} and~\ref{algo:smooth_BMD} have been implemented in both MATLAB and Python. The MATLAB implementation is available on the MathWorks File Exchange at \footnote{\href{https://www.mathworks.com/matlabcentral/fileexchange/180648-smooth-bloch-messiah-decomp-of-conj-sympl-matrix-function}{https://www.mathworks.com/matlabcentral/fileexchange/180648-smooth-bloch-messiah-decomp-of-conj-sympl-matrix-function}}, 
while the Python version can be found on GitHub at \footnote{\url{https://github.com/apulian/smooth-BMD/}}.

\begin{remark} We highlight that Algorithm \ref{algo:smooth_BMD} follows the general philosophy of predictor-corrector methods. At each step $\omega_k$, the prediction (step 2.1) is given by a Bloch-Messiah decomposition at $\omega_k$ computed via Algorithm \ref{algo:Bloch-Messiah}, while the correction (step 2.3) is performed by adjusting the unitary factors at $\omega_k$ to be as close as possible to those at $\omega_{k-1}$ via post-multiplication by a diagonal phase matrix. We point out that:
\begin{itemize}
    \item[i)] In step 2.1, the step size $h$ is chosen adaptively to ensure that the distances $\norm{U_{k-1}-U_k}_\mathrm{F}$, $\norm{V_{k-1}-V_k}_\mathrm{F}$, and $\norm{D_{k-1}-D_k}_\infty$ remain below (but sufficiently close to) a user-specified threshold;
    \item[ii)] The minimization problem in step 2.2 can be reformulated and efficiently solved as a special orthogonal Procrustes problem.
\end{itemize}
Finally, we remark that one can easily adapt Algorithm \ref{algo:smooth_BMD} to obtain a reduced \jmvbmd, where only a selected portion of dominant singular values and corresponding singular vectors is computed. This feature lowers the overall computational complexity of the algorithm and has been used in the experiments reported in Section~\ref{subsec:large}.\end{remark}

\begin{remark}\label{remark:caution} We stress that computing an ABMD (let alone a \jmvbmd) is not simply a matter of collecting BMDs computed at different parameter values. Failing to recognize this fact can lead to incorrect results, where, for instance, avoided crossings of singular values may be mistaken for true crossings.
\end{remark}


\FloatBarrier
\section{Physical examples}\label{sec:phys_ex}

In this final section, we apply the theoretical results and algorithms from Section \ref{sec:LinAlg} to two physical systems. We begin with a small system that is small enough to be fully described while still exhibiting all the features of generic systems. We then conclude with a higher-dimensional one.

\subsection{A four-mode system}\label{subsec:four}
Here we consider a four-mode system whose mode-hopping and pair-generation terms can be experimentally controlled. This could be implemented in several physical platforms, for example in arrays of coupled nonlinear photonic cavities~\cite{Cao2016,Ye2023,Zhao2023,Jabri2024,Ravets2025}, engineered multimode parametric oscillators~\cite{Patera2012} or $\chi^{(3)}$ nonlinear microresonators~\cite{Bensemhoun2024}.
Thus, we consider $S$ as in \eqref{S} and $\mathcal{M}$ as in \eqref{eq:M_cali}, and we start by choosing
\begin{align}
G/\gamma&=
\left[
\begin{array}{cccc}\label{eq:G codim3 n4}
1.35 & 0.08 & 0.65 & 0.08
\\
0.08 & 1.25 & 0.08 & 0.65
\\
0.65 & 0.08 & 1.35 & 0.08
\\
0.08 & 0.65 & 0.08 & 1.25
\end{array}
\right],
\\
F/\gamma&=
\left[
\begin{array}{cccc}
0.08 & 0.25 & 0.1 & 0.5+i 0.05
\\
0.25 & 0.1 & 0.5 + i 0.05 & 0.1
\\
0.1 & 0.5+i 0.05 & 0.1 & 0.25
\\
0.5 + i 0.05 & 0.1 & 0.25 & 0.08
\end{array}\label{eq:F codim3 n4}
\right],\\
\Gamma/\gamma &= \diag(1,1.5,1,1.5,\ldots,1,1.5),
\end{align}
where all quantities have been normalized to the damping rate of the first mode, say $\gamma$.

\begin{figure}[t!]
    \centering
    \includegraphics[width=0.45\textwidth]{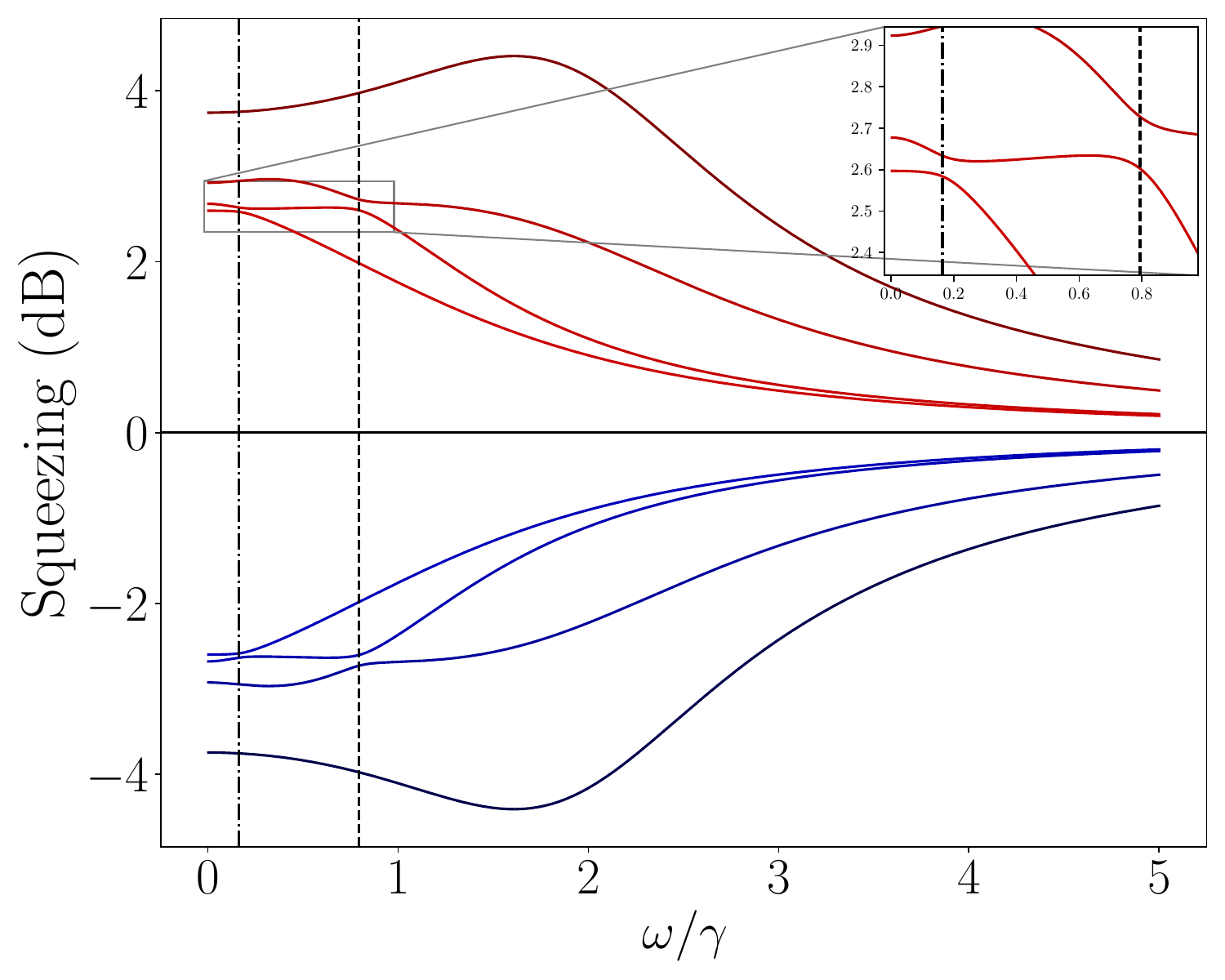}
    \caption{Case $n=4$ and $\text{codim}=3$. Squeezing, $d_{n+j}^{-2}(\omega)$ and anti-squeezing, $d_i^2(\omega)$ (with $j\in\{1,\ldots,4\}$) spectra, in dB.  The specific choice for the parameters gives an avoided crossing between the second ($j=1$) and third ($j=2$) squeezing and anti-squeezing spectra around $\omega=0.794\gamma$ (marked by the vertical black-dashed line) and an avoided crossing between the third ($j=2$) and forth ($j=3$) squeezing spectra around $\omega=0.164\gamma$ (marked by the vertical black-dot-dasjhed line). The inset shows a magnified view of the two avoided crossing regions in the anti-squeezing spectra.}
    \label{fig:n=4 codim=3 curves}
\end{figure}
\begin{figure}[t!]
  \centering

  \begin{subfigure}[t]{\columnwidth}
    \centering
    \includegraphics[width=\linewidth]{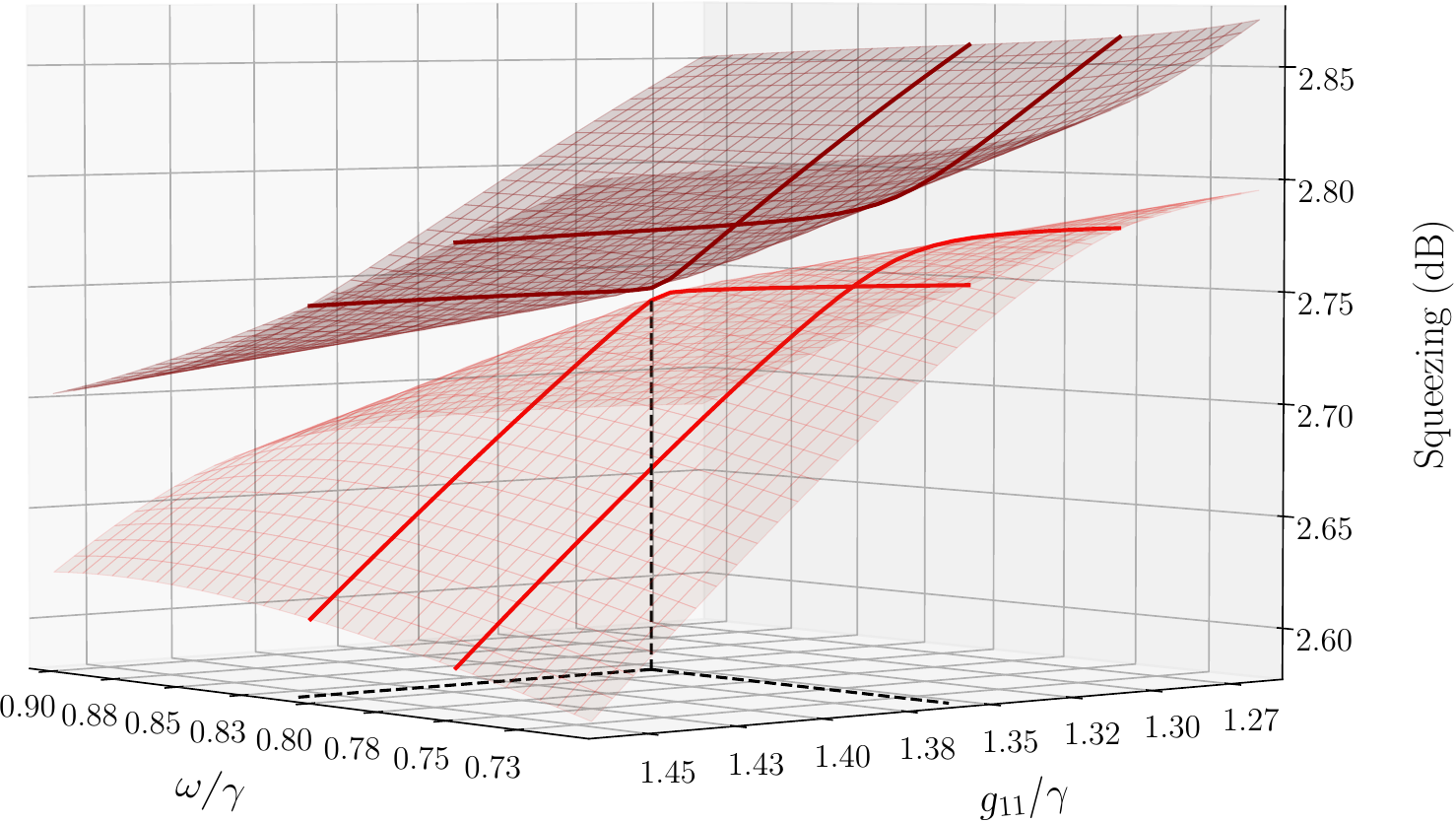}
    \caption{Case of avoided crossing.}
    \label{fig:subfig1}
  \end{subfigure}

  \vspace{0.5em}

  \begin{subfigure}[t]{\columnwidth}
    \centering
    \includegraphics[width=\linewidth]{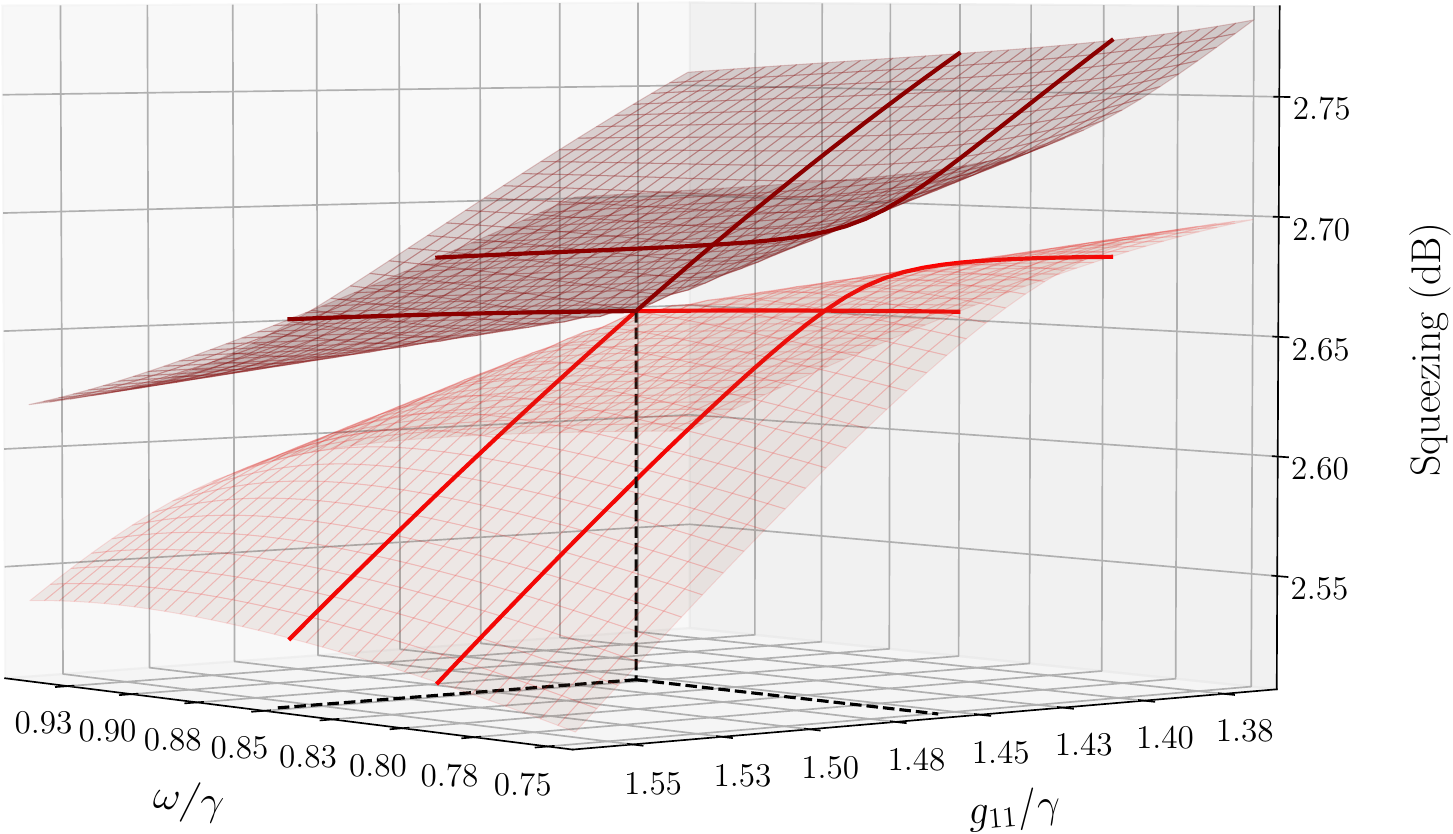}
    \caption{Case of a point of degeneracy.}
    \label{fig:subfig2}
  \end{subfigure}

  \caption{Anti-squeezing spectra in dB ($d_1^2(\omega,g_{12})$ and $d_2^2(\omega,g_{12})$) for the case $n=4$, $\mathrm{codim}=3$. (a) For $g_{22}=1.25\gamma$: the surfaces of anti-squeezing avoid crossing through the entire plane $(\omega,g_{11})$. (b) For $g_{22}\approx1.12\gamma$: the surfaces of squeezing touch at one point of degeneracy in the plane $(\omega,g_{11})$.}
  \label{fig:n3-codim3}
\end{figure}
First, we computed the \jmvabmd\ of $S(\omega)$ over the interval $\omega/\gamma\in[0,5]$. Figure \ref{fig:n=4 codim=3 curves} shows the anti-squeezing (shades of red) and squeezing (shades of blue) spectra obtained from the singular values. Let us consider the singular values that give the anti-squeezing spectra, i.e. $d_j(\omega)$ for $j=1,\ldots,4$. Owing to the symmetry of the singular value spectrum with respect to the shot noise level (0 dB), similar features also appear in the squeezing spectra. Notably, $d_2(\omega)$ and $d_3(\omega)$ appear to approach a crossing near $\omega\approx0.8\gamma$ before veering away from each other. 
\begin{figure}[t!]
    \centering
    \includegraphics[width=0.45\textwidth]{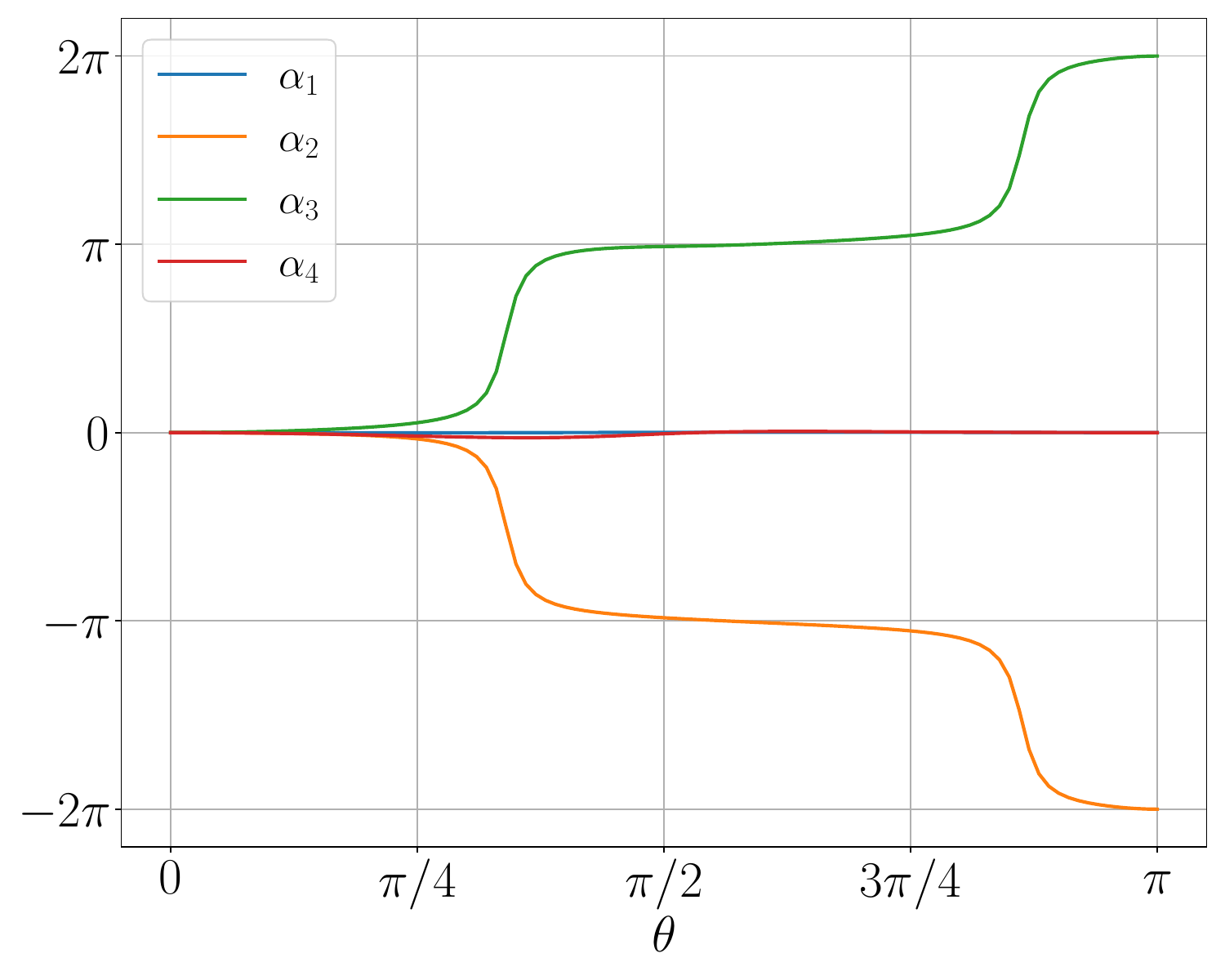}
    \caption{Continuous Berry phases $\alpha_j(\theta)$ for $j = 1, \ldots, 4$ and $\theta \in [0, \pi]$ associated to the corresponding singular vectors of \eqref{eq:Sfun_3param} computed along the parallels of the sphere of radius $0.5$, centered at $(\omega, g_{11}, g_{22}) = (0.8, 1.35, 1.25)$. Note that $\alpha_2$ and $\alpha_3$ do not return to zero, indicating that $d_2$ and $d_3$ experience a degeneracy within the region enclosed by the sphere.}
    \label{fig:n=4 codim=3 berry phases}
\end{figure}
A similar behavior is observed for $d_3(\omega)$ and $d_4(\omega)$ near $\omega\approx0.15\gamma$. These so-called \emph{avoided crossings} are typical in the spectra of matrix-valued functions when the number of varying parameters is smaller than the codimension of the degeneracies. Specifically, since the codimension of degeneracies is 3, an actual degeneracy should not be expected when varying only one or two parameters. This is illustrated in Figure~\ref{fig:subfig1}, which shows two singular-value surfaces undergoing an avoiding crossing. We also note that, due to the presence of avoided crossings with $d_4$ and $d_2$, the singular value $d_3$ takes an almost flat structure through the spectral interval between $0.15\gamma$ and $0.8\gamma$.

The avoided crossings in Figure~\ref{fig:n=4 codim=3 curves} should be interpreted as an indication of the possible presence of nearby degeneracies, which could be observed by freeing two additional parameters. This is confirmed by considering the matrix function
\begin{equation}\label{eq:Sfun_3param}
    (\omega,g_{11},g_{22})\mapsto S(\omega,g_{11},g_{22})\in\C^\tnxtn,
\end{equation}
which is conjugate symplectic for all values of $(\omega,g_{11},g_{22})$.
Applying the topological test of Theorem \ref{thm:codim3detect} to the sphere of radius $0.5\gamma$ centered at $(\omega,g_{11},g_{22})=(0.8\gamma, 1.35\gamma, 1.25\gamma)$, we obtain the continuous Berry phases shown in Figure \ref{fig:n=4 codim=3 berry phases}, which show that $\alpha_j(\pi)\ne\alpha_j(0)$ for $j=2,3$. This confirms that $d_2$ and $d_3$ undergo a degeneracy inside the region enclosed by the sphere. To further refine the degeneracy's location, we minimized the difference
\begin{equation*}
    d_2(\omega,g_{11},g_{22})-d_3(\omega,g_{11},g_{22}).
\end{equation*}
using MATLAB's simplex-based method \texttt{fminsearch} and using the center of the sphere as the initial seed. This yielded the approximate degeneracy location:
\begin{equation*}
    \begin{bmatrix}
        \omega \\ g_{11} \\ g_{22}
    \end{bmatrix}\approx
    \begin{bmatrix}
    0.84582091\gamma \\
     1.4623127\gamma \\
     1.1179756\gamma
    \end{bmatrix}
\end{equation*}
with difference $d_2-d_3\approx3.4\times10^{-9}$. Choosing $g_{22} \approx 1.1179756\gamma$, a degeneracy between the singular value surfaces $d_2$ and $d_3$ becomes apparent in the $(\omega, g_{11})$ plane, as illustrated in Figure~\ref{fig:subfig2}. Figures displaying the entries of the singular vectors and showing their rapid variations near the avoided crossings are presented in Appendix~\ref{Appendix:singular vectors}. A detailed discussion of this phenomenon will be provided in the next subsection.
\begin{figure}[t!]
    \centering
    \includegraphics[width=0.45\textwidth]{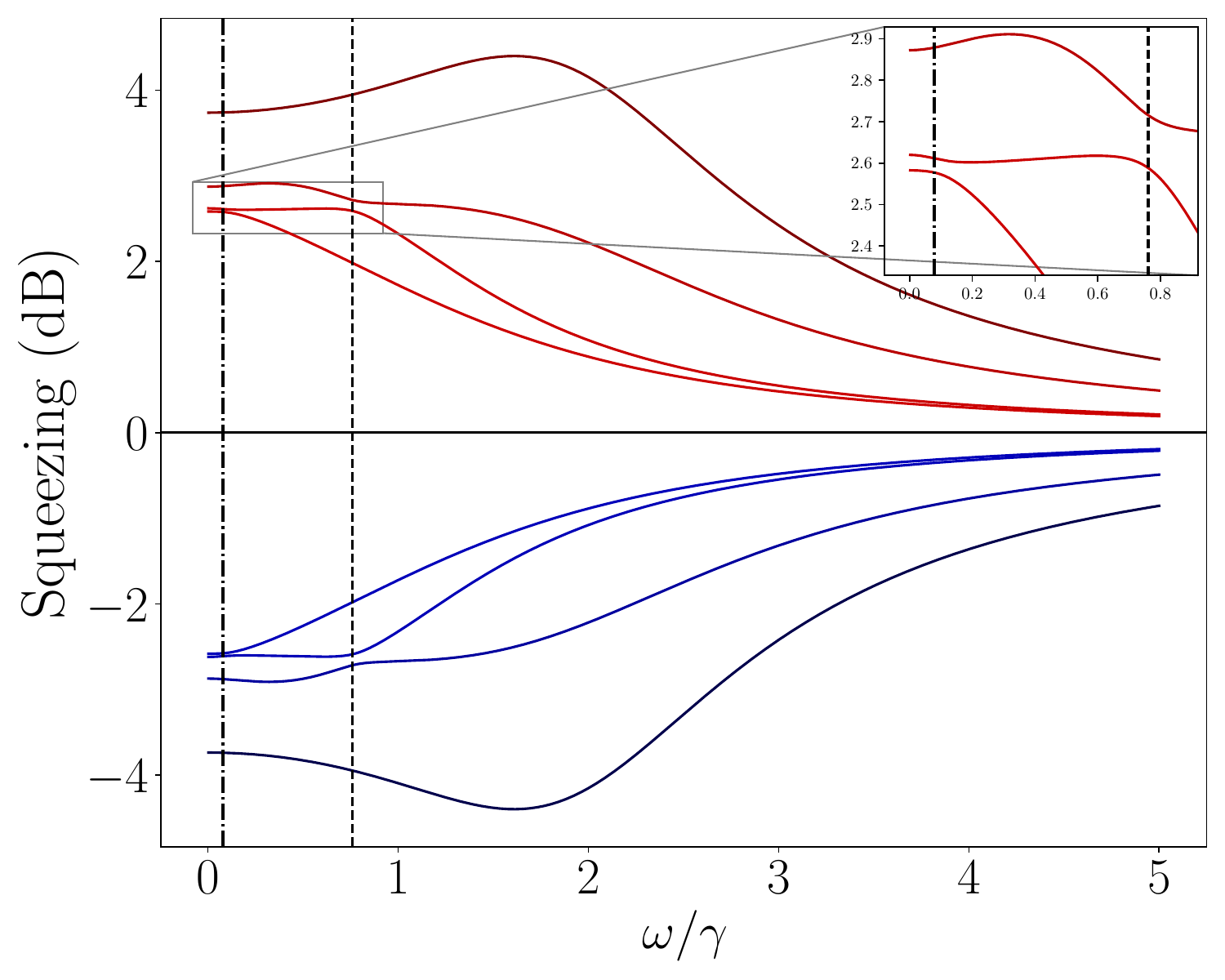}
    \caption{Case $n=4$ and $\text{codim}=2$. Squeezing, $d_{n+j}^{-2}(\omega)$ and anti-squeezing, $d_j^2(\omega)$ (with $j\in\{1,\ldots,4\}$) spectra, in dB. The specific choice for the parameters gives an avoided crossing between the second ($j=1$) and third ($j=2$) squeezing and anti-squeezing spectra around $\omega=0.761\gamma$ (marked by the vertical black-dashed line) and an avoided crossing between the third ($j=2$) and forth ($j=3$) squeezing spectra around $\omega=0.079\gamma$ (marked by the vertical black-dot-dashed line). The inset shows a magnified view of the two avoided crossing regions in the anti-squeezing spectra.}
    \label{fig:n=4 codim=2 curves}
\end{figure}
\begin{figure}[t!]
    \centering
    \includegraphics[width=1\columnwidth]{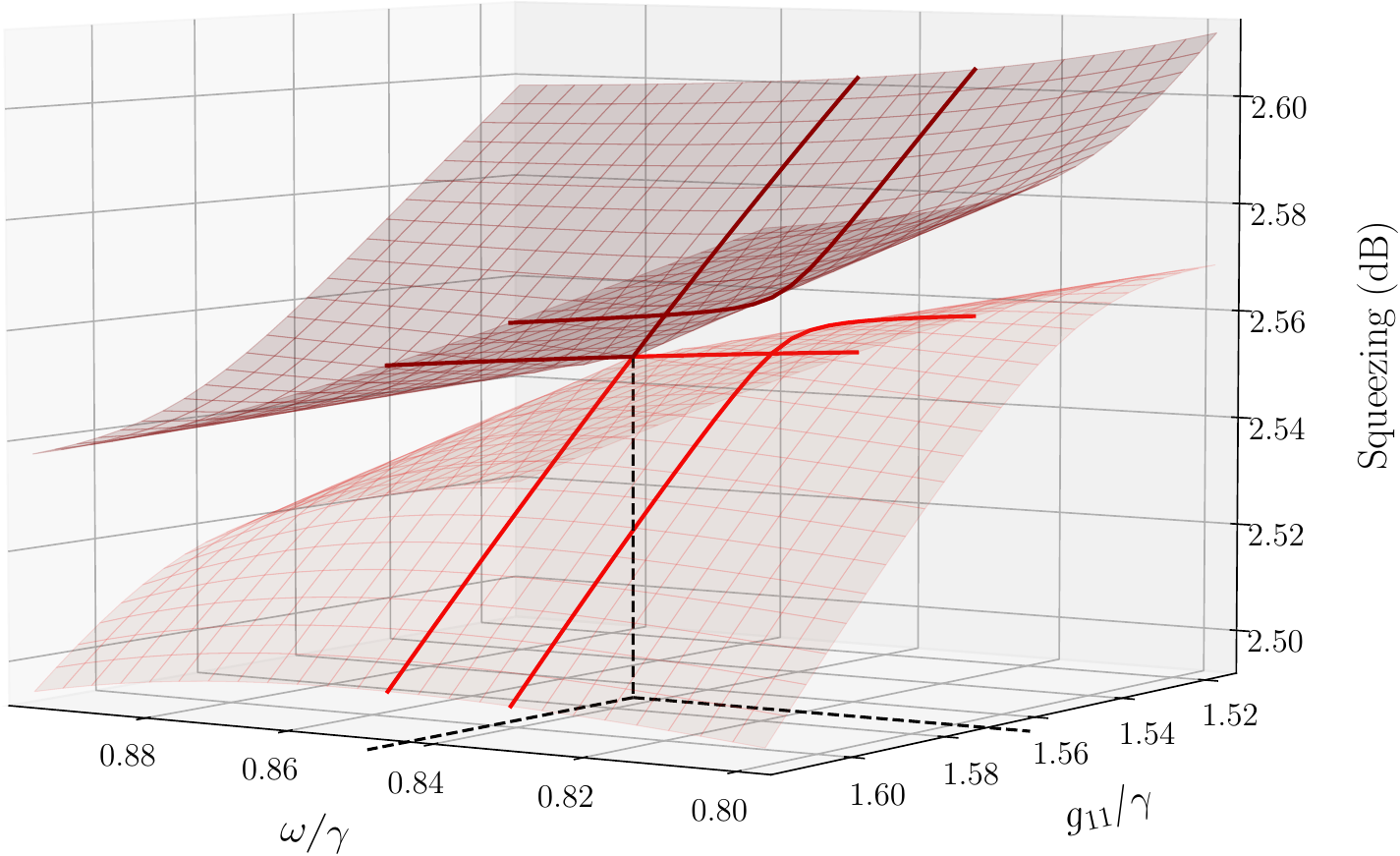}
    \caption{Anti-squeezing spectra ($d_1^2(\omega)$ and $d_2^{2}(\omega)$), in dB, for the case $n=4$ and $\text{codim}=2$. The surfaces of squeezing touch at one point of degeneracy in the plane $(\omega,g_{11})$ corresponding to $\omega\approx0.847\gamma$ and $g_{11}\approx1.56\gamma$.}
    \label{fig:n=4 codim=2 surfs}
\end{figure}


We also considered the situation where the couplings are fully real, replacing $F$ of~\eqref{eq:F codim3 n4} with its real part. Computing the four curves of singular values, we obtained Figure \ref{fig:n=4 codim=2 curves}, which is qualitatively similar to Figure \ref{fig:n=4 codim=3 curves} and shows two avoided crossings. When  investigating the existence of a degeneracy near the second avoided crossing as we did previously, we observed that the Berry phases considered in Theorem \ref{thm:codim3detect}, and computed via Algorithm \ref{algo:smooth_BMD}, took only values that are multiples of $\pi$. This behavior contrasts with the codimension 3 case, where, in general, $\alpha_j$ varies smoothly across a continuous range of values when $\theta$ is varied (see Figure~\ref{fig:n=4 codim=3 berry phases}, particularly $\alpha_j$ for $j=2, 3$). This suggests that the codimension of degeneracies may be lower than expected~\footnote{For instance, the Berry phase takes only values that are multiples of $\pi$ when the matrix function is real symmetric, in which case the codimension of degeneracies is in fact 2, see \cite{DiPu_realSVD,DIPU_MathCompSym2008}.}. Indeed, this turns out to be the case. It is easy to see that, when both $F$ and $G$ are real, the matrix $S(\omega)$ satisfies
\begin{equation*}
    S(\omega)\begin{bmatrix}
        0 & I \\ I & 0
    \end{bmatrix}=
    \left(
    S(\omega)\begin{bmatrix}
        0 & I \\ I & 0
    \end{bmatrix}
    \right)^T;
\end{equation*}
in other words, $S(\omega)\begin{bmatrix} 0 & I \\ I & 0 \end{bmatrix}$ is complex symmetric for all values of $\omega$. It is known (see \cite{DiPaPu_takagi}) that the degeneracy of singular values for complex symmetric matrices has codimension $2$, and that the right tool to use in this case is the Takagi decomposition. The same result on the codimension must be true also for $S(\omega)$, since it has the same singular values of $S(\omega)\begin{bmatrix} 0 & I \\ I & 0 \end{bmatrix}$.
In Appendix \ref{sec:takagi vs jointMVD}, we show that the joint-MVD (\cite{DIPU_LAA2024}) of a complex symmetric matrix function is ``essentially'' a smooth Takagi decomposition. This fact allows us to apply the results in \cite{DiPaPu_takagi} to infer the presence of a point of degeneracy by introducing one extra degree of freedom and computing the joint-MVD around loops in 2-parameters space. In \cite{DiPaPu_takagi}, the authors show that a change of sign of a singular vector of a smooth Takagi decomposition around a loop in two-parameters' space signifies the presence of a point of degeneracy for the corresponding singular value inside the loop. We considered the $\blam$-symplectic function
\begin{equation*}
    (\omega,g_{11})\mapsto S(\omega,g_{11})\in\C^\tnxtn,
\end{equation*}
and computed the \jmvabmd\ of $S$ around the circle of radius $0.25\gamma$ centered at $(0.75\gamma, 1.35\gamma)$. A comparison between the unitary factors at the beginning and the end of the loop yielded
\begin{equation*}
    U(0)^\dag U(1) = V(0)^\dag V(1)=
\begin{bmatrix}
1 & 0 & 0 & 0 \\
0 & -1 & 0 & 0 \\
0 & 0 & -1 & 0 \\
0 & 0 & 0 & 1
\end{bmatrix}.
\end{equation*}
This, according to \cite{DiPaPu_takagi}, signals the presence of a degeneracy inside the circle for the pair $(d_2,d_3)$. Minimizing the difference
\begin{equation*}
    d_2(\omega,g_{11})-d_3(\omega,g_{11})
\end{equation*}
through MATLAB's \texttt{fminsearch}, using the center of the circle as initial seed, yielded the approximate degeneracy location:
\begin{equation*}
    \begin{bmatrix}
        \omega \\ g_{11}
    \end{bmatrix}\approx
    \begin{bmatrix}
    0.84764778 \gamma \\
    1.5643801 \gamma
    \end{bmatrix}.
\end{equation*}
with difference $d_2-d_3\approx3.2\times10^{-9}$.
\begin{remark}\label{rem:degeneracies only in dynamics} 
For every value of the parameters $g_{11}, g_{22}$ at which the spectrum of the transfer function $S$ displayed a degeneracy, we computed the corresponding spectrum of the Hamiltonian and found no degeneracy. This shows that the spectral degeneracies under consideration in this work manifest themselves only ``through the driven-dissipative quantum dynamics''.
\end{remark}
\begin{figure}[t!]
\includegraphics[width=\columnwidth]{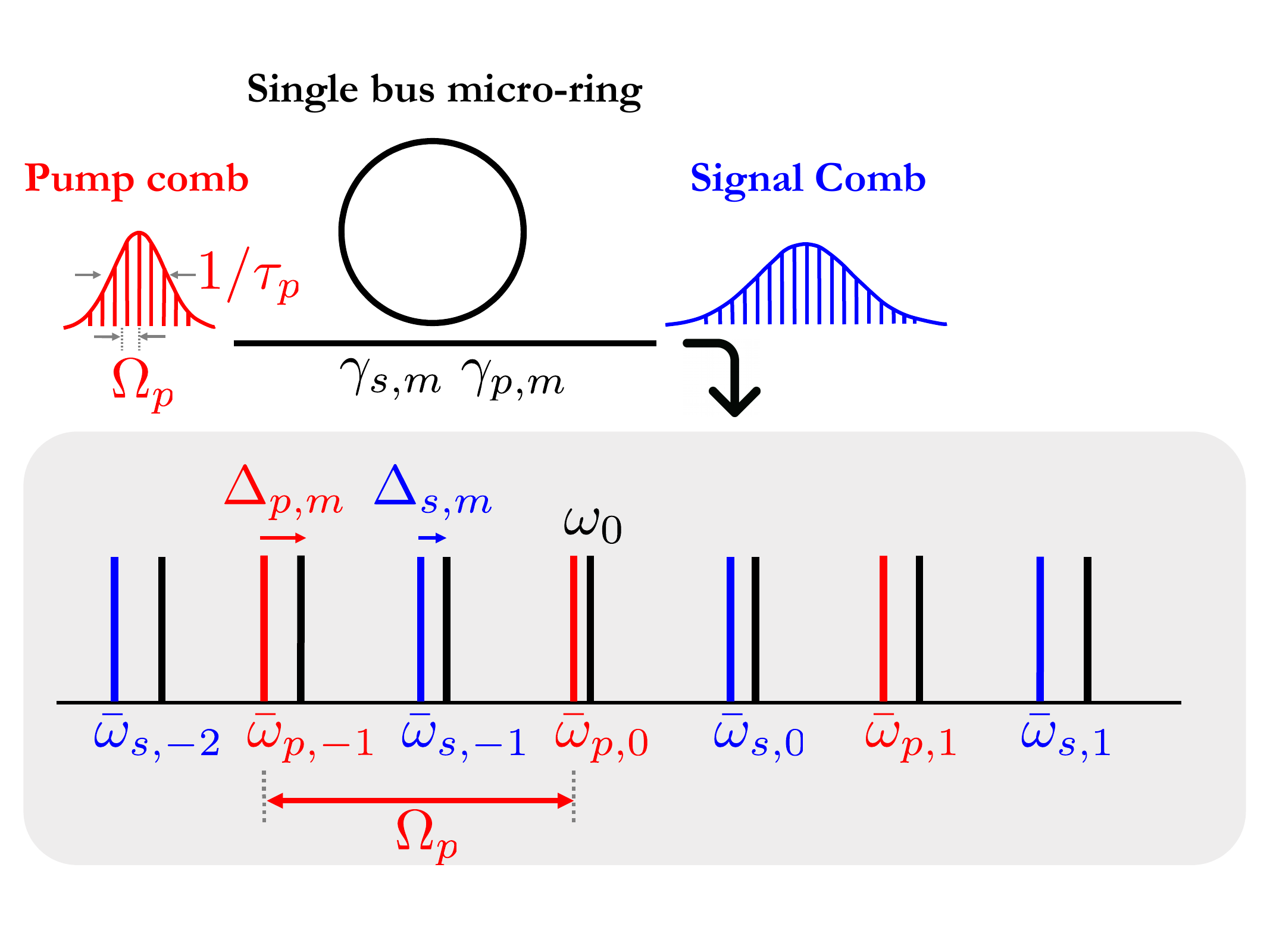}
\caption{Schematic of a synchronously pumped micro-ring.
The pump frequency comb has carrier $\omega_p = \bar{\omega}_{p, 0}$, pulse duration $\tau_{\mathrm{p}}$ and
repetition rate $\Omega_{\mathrm{p}}$ approximately equal to the double of the cavity FSR so that the pump spectral components match one cavity resonance out of two, 
with mode-dependent detuning $\Delta_{\mathrm{p},m}$.
Coupling losses are given by $\gamma_{\mathrm{f},m}$, with f$=\{$p,s$\}$.
}\label{fig:spopo}
\end{figure}

\subsection{The pulsed microring resonator in synchronously pumped regime}\label{subsec:large}

As a case of study with large dimension, we consider an optical nonlinear cavity whose dynamics is highly multimode.
This system, represented in Figure~\ref{fig:spopo}, is the same studied in~\cite{Gouzien2023} and consists of a microring resonator with a $\chi^{(3)}$ nonlinearity, coupled to a single straight injection waveguide (single-bus device) and driven by a comb of equally spaced spectral components
\begin{align}
\bar{\omega}_{\mathrm{p},j}=\omega_{\mathrm{p}}+j\Omega_{\mathrm{p}}
\label{omegap bar}
\end{align}
($j\in\mathbb{Z}$), centered around the frequency $\omega_\mathrm{p}$.
In order to have distinct signal and pump modes and guarantee a synchronous pumping regime, the driving comb free spectral range (FSR), $\Omega_{\mathrm{p}}$, is taken to be
about 2 times the cavity average FSR. The cavity bare resonances are given by:
\begin{equation}\label{omega_m}
\omega_{j}=
\omega_0+\sum_{j\geq 1}\frac{\Omega_j}{j!}m^j.
\end{equation}
The reference label $j=0$ indicates the resonance whose frequency approximately matches the pump carrier, $\omega_0\approx\omega_\mathrm{p}$.
The first order parameter $\Omega_1=c/(n_g R_{\mathrm{eff}})$ gives the average cavity FSR in terms of the speed of light in vacuum
$c$, the group index $n_g$, and the ring effective radius $R_{\mathrm{eff}}$.
The parameter $\Omega_{2}=-(n'_g c^ 2)/(n_g^ 3 R_{\mathrm{eff}}^ 2)$ accounts for second-order dispersion effects via the frequency derivative $n'_g$ together with higher-order dispersion terms $\Omega_{k>2}$.

Since $\Omega_\mathrm{p}\approx2\Omega_1$, the pump injection components approximately match the cavity resonances of even order so that
\begin{align}
\bar{\omega}_{\mathrm{p},j}\approx\omega_{\mathrm{p}}+2j\Omega_1.
\end{align}
Due to dispersion, their detuning with respect to even cavity resonances $\Delta_{\mathrm{p},j}=\omega_{2j}-\bar{\omega}_{\mathrm{p},j}$ changes with $j$.
Frequency signal modes are generated by FWM at frequencies
\begin{align}
\bar{\omega}_{\mathrm{s},j}=\omega_{\mathrm{p}}+(2j+1)\frac{\Omega_{\mathrm{p}}}{2}
\label{omega s bar}
\end{align}
and can thus be unequivocally distinguished from the pump (``s'' stands for ``signal'').
\begin{figure}[t!]
\centering
\includegraphics[width=1.\columnwidth]{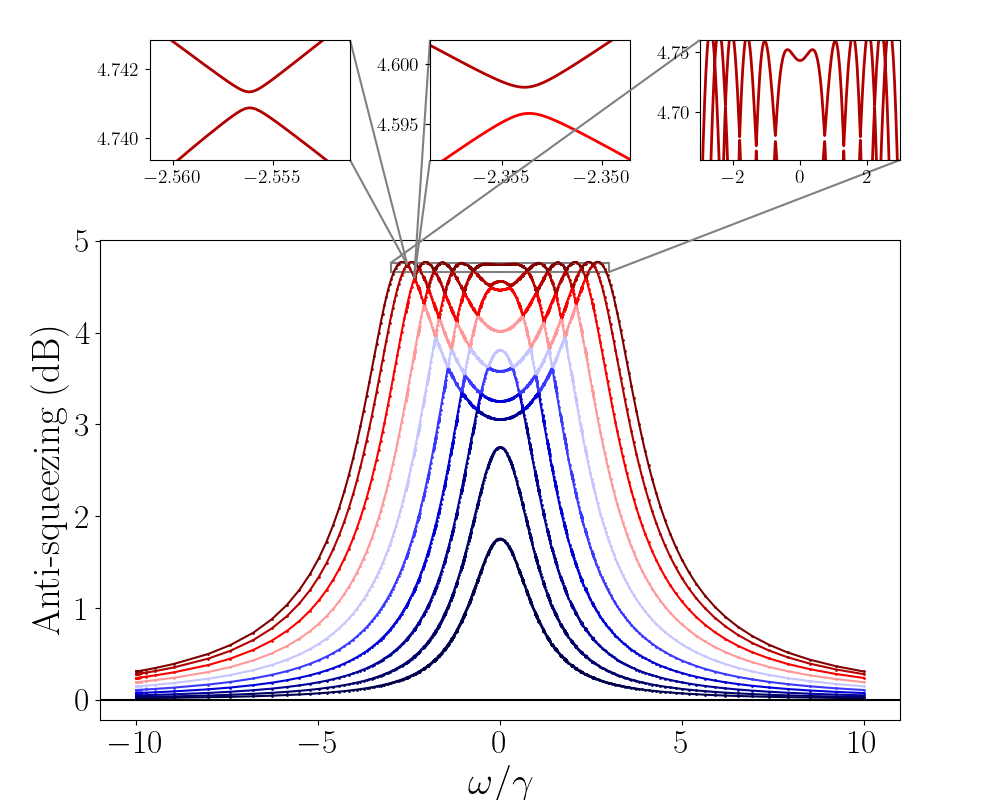}
\caption{First ten anti-squeezing spectra, $d_j^2(\omega)$ in dB (for $j=\{1,\ldots,10\}$), for the case of pulsed microring resonator in synchronously pumped regime with $n=51$. The leftmost inset shows an avoided crossing between the first and second anti-squeezing spectra around $\omega\approx -2.556\gamma$. The central inset shows an avoided crossing between the second and third anti-squeezing spectra around $\omega\approx -2.354\gamma$. The rightmost inset shows the top part of the the first and second anti-squeezing spectra.}
\label{fig:eval}
\end{figure}
\begin{figure*}[t!]
\centering
\includegraphics[width=1\textwidth]{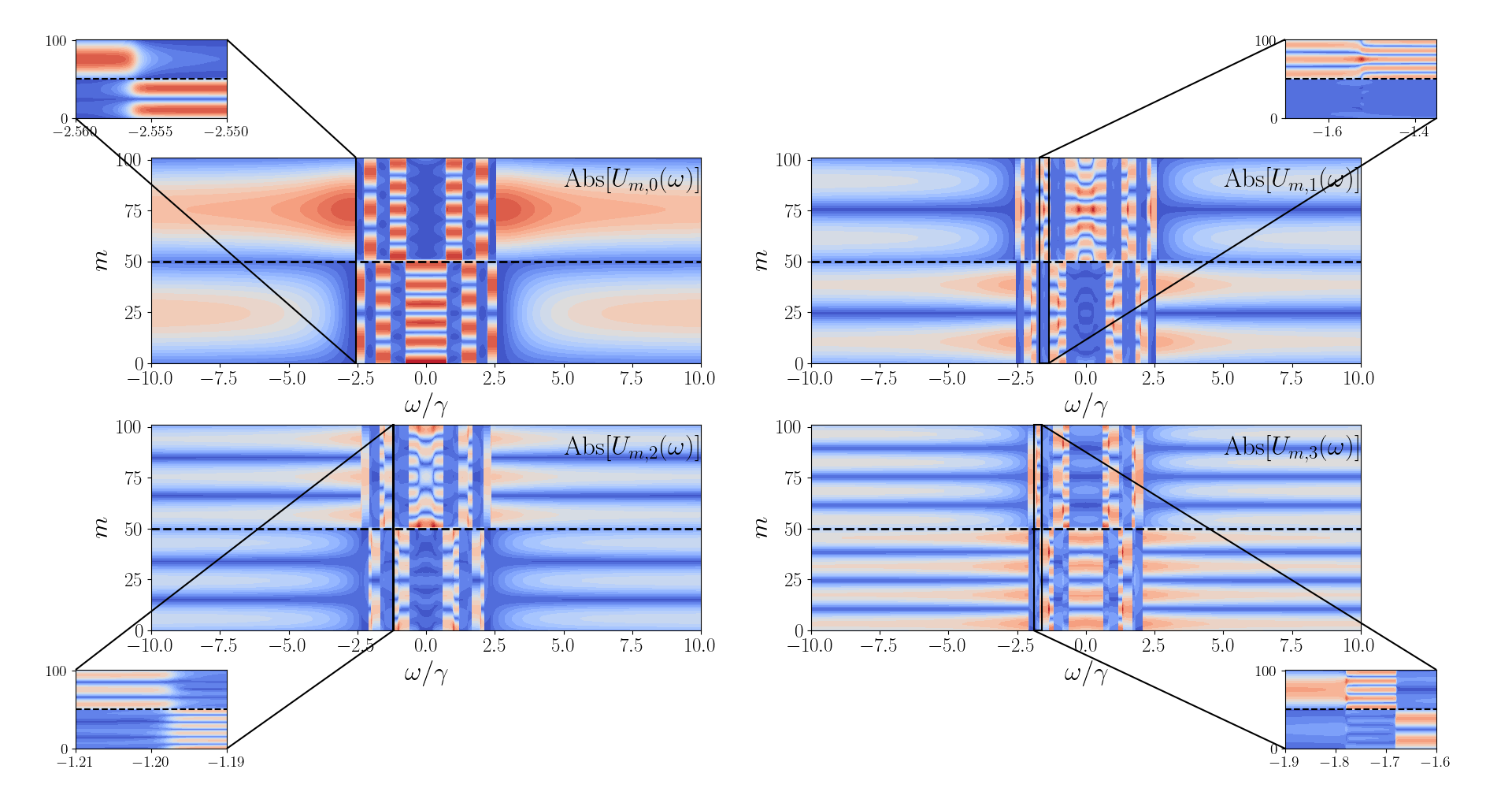}
\caption{First four most anti-squeezed morphing supermodes ($j = {1,\ldots,4}$) for a pulsed microring resonator in the synchronously pumped regime with $n = 51$. The enlarged views show the rapid (yet smooth) rotation of two supermodes near an avoided crossing.}\label{fig:evec}
\end{figure*}
They are in general detuned by $\Delta_{\mathrm{s},j}=\omega_{2j+1}-\bar{\omega}_{\mathrm{s},j}$ with respect to the odd cavity resonances. By performing a Taylor expansion up to second order, the detuning can be written as
\begin{equation}
    \Delta_{\mathrm{s},j}
    \approx
    \Delta_{\mathrm{s},0}
    +\Delta\Omega\left(2j+1\right)
    +\frac{\Omega_2}{2!}(2j+1)^2,
\end{equation}
where $\Delta_{\mathrm{s},0}=\omega_0 - \omega_{\mathrm{p}}$ represents the detuning between the cavity resonance $j=0$ and the carrier of the injection, and $\Delta\Omega=\Omega_1-\Omega_{\mathrm{p}}/2$ represents the mismatch with respect to perfect synchronization.
Below threshold the signal modes remain macroscopically empty but their fluctuations undergoes a multimode dynamics that can generate multimode vacuum squeezed state. By collecting the signal quadratures, $\hat{\bm{x}}_{\mathrm{s}}(t)=(\ldots,\hat{x}_{\mathrm{s},-1},\hat{x}_{\mathrm{s},0},\hat{x}_{\mathrm{s},1},\ldots)^\transp$ and $\hat{\bm{y}}_{\mathrm{s}}(t)=(\ldots,\hat{y}_{\mathrm{s},-1},\hat{y}_{\mathrm{s},0},\hat{y}_{\mathrm{s},1},\ldots)^\transp$, in the column vector $\hat{\bm{R}}(t)=(\hat{\bm{x}}_{\mathrm{s}}(t)|\hat{\bm{y}}_{\mathrm{s}}(t))^\transp$, one can see that the corresponding linearized dynamics is described by eq.~\eqref{eq:quantum_lang_r} where the coupling matrix $\mathcal{M}$, eq.~\eqref{eq:M_cali}, with
\begin{align}
F_{j,\ell} &= g \sum_{m} {\expval{\hat{p}_{j-m+\ell+1}}} {\expval{\hat{p}_m}}, \label{F} \\
G_{j,\ell} &= \Delta_{\mathrm{s},j}\delta_{j,\ell}
	+ g \sum_m 2 {\expval{\hat{p}_{j+m-\ell}}} {\expval{\hat{p}_m}}^*, \label{G}
\end{align}
with $g$ the nonlinear coupling, $\langle\hat{p}_j\rangle$ the steady state solutions of the classical nonlinear equation for the intracavity pump modes (for further details see~\cite{Gouzien2023}) and $\delta_{j,\ell}$ the Kronecker delta.
We consider the case where the injection is engineered so that the pump modes are macroscopically populated with a (real) Gaussian distribution
\begin{equation}\label{eq:pjs}
\left\langle \hat{p}_j \right\rangle
=
A_0\, 
\mathrm{e}^{-(j-j_0)^2/2\sigma^2},
\end{equation}
where $j_0$ is the pump mode order defining the center of the Gaussian distribution and $\sigma$ its width measured in units of $\Omega_1$. We choose $n = 51$, $\Gamma = \gamma\,I_{2n}$ (i.e., all modes undergo the same damping rate). We consider perfect resonance $\Delta_{\mathrm{s},0} = 0$, perfect synchronization $\Delta\Omega = 0$, and dispersion-compensated operation ($\Omega_2 = 0$). We take $\sigma = 3$ and center equal to $j_0 = (n - 1)/2$.
It is important to note the impact of the parameter $A_0$ on the spectrum of singular values: increasing $A_0$ shifts all singular values upward, while decreasing it shifts them downward. This effect not only raises the entire spectrum but also tends to compress the singular values toward the upper end. This behavior can be exploited to tailor in a rudimentary, but simple, way the structure of the squeezing spectra. Indeed, the higher the value of $A_0$, the greater the number of spectra that come close to each other and undergo avoided crossings at multiple points (see Figure~\ref{fig:eval}). As a result, the largest spectra are constrained to oscillate within a narrow range of values, thus giving origin to an almost flat-band spectrum. This feature is important for the quality of de-Gaussification protocols of squeezed states, as pointed out by Asavanant \textit{et al.} in~\cite{Asavanant2017}.

In the following, we set $A_0 = 0.4\sqrt{\gamma/g}$. We then consider the \jmvabmd\ of $S(\omega)$ with $\omega \in [-10\gamma, 10\gamma]$, and compute its first 10 modes (i.e., its 10 largest singular values and the associated singular vectors).
Figure~\ref{fig:eval} shows the outcome of the computation. 
In particular, it shows that the first eight most anti-squeezed modes (and, by symmetry, the same holds for the squeezing spectra) go through several avoided crossings. The first (i.e. largest) anti-squeezed spectrum experiences ten avoided crossings with the second one, forcing it to oscillate within a very narrow interval and making it nearly flat-band. Figure~\ref{fig:evec} shows the first four most anti-squeezed morphing supermodes. Focusing on the first supermode (top-left in Figure~\ref{fig:evec}), we observe that at the first avoided crossing, located at $\omega \approx -2.556\gamma$, it rapidly yet smoothly rotates (see the enlarged view in the top-left panel) and exchanges with the second morphing supermode (top-right in Figure~\ref{fig:evec}). It then undergoes a second avoided crossing located at $\omega \approx -2.237\gamma$ and exchanges with the third supermode. A similar behavior is observed at other avoided crossings -- all occurring within the interval $\omega \in [-2.5\gamma, 2.5\gamma]$ -- where the first supermode smoothly rotates and exchanges with the corresponding supermode.

We run through several avoided crossings, and decided to investigate the presence of a point of degeneracy near the leftmost one, which involves the pair $(d_1, d_2)$. We note that, since the pump modes are real valued [see \eqref{eq:pjs}], also in this case the codimension of the degeneracies is 2. Repeating what we have done in the previous section, we introduced the additional degree of freedom $\Delta_{\mathrm{s},0}$ and computed the Berry phase associated to all the singular values for the conjugate $\blam$-symplectic matrix-valued function
\begin{equation*}
    (\omega,\Delta_{\mathrm{s},0})\mapsto S(\omega,\Delta_{\mathrm{s},0})\in\C^\tnxtn.
\end{equation*}
around the circle of radius $0.1\gamma$ centered at $(-2.556\gamma,0)$. Comparing the unitary factors at the beginning and end of the circle yielded
\begin{equation*}
    U(0)^\dag U(1) = V(0)^\dag V(1)=
\begin{bmatrix}
-1 & 0 & 0 &  \cdots & 0 \\
0 & -1 & 0 & \cdots & 0 \\
0 & 0 & 1 & \cdots & 0 \\
\vdots & \vdots & \vdots & \ddots & \vdots \\
0 & 0 & 0 & \cdots & 1
\end{bmatrix}.
\end{equation*}
This signals the presence of a degeneracy inside the circle for the pair $(d_1,d_2)$. Minimizing the difference
\begin{equation*}
    d_1(\omega,\Delta_{\mathrm{s},0})-d_2(\omega,\Delta_{\mathrm{s},0})
\end{equation*}
through MATLAB's \texttt{fminsearch}, using the center of the circle as initial seed, yielded the approximate degeneracy location:
\begin{equation*}
    \begin{bmatrix}
        \omega \\ \Delta_{\mathrm{s},0}
    \end{bmatrix}\approx
    \begin{bmatrix}
    -2.4943524 \gamma \\ -0.057548294 \gamma
    \end{bmatrix}.
\end{equation*}
with difference $d_1-d_2\approx1.9\times10^{-9}$.

\section{Conclusions}
In this work, we have developed a detailed framework for the smooth decomposition of matrix-valued functions arising in driven-dissipative multimode bosonic quantum systems -- specifically, $\omega$- and $\blam$-symplectic transformations -- with a focus on their numerical computation in a way that preserves their algebraic structures. Importantly, we have clarified that a smooth decomposition is not merely a collection of pointwise or static decompositions, but rather a global object with nontrivial geometric and topological properties. In this context, we have discussed the role of degeneracies and the associated codimension, which play a crucial role in shaping the behavior of the decomposition.
These theoretical insights have been applied to a class of physically relevant problems in multimode quantum open systems, where we analyzed the emergence of both avoided and genuine crossings in the squeezing spectrum. We showed that avoided crossings can be exploited for the engineering of flat-band squeezed states, and identified the presence of genuine crossings associated with nontrivial topological Berry phases in the singular vectors. We also highlighted the hypersensitivity of the singular vectors in proximity to these degeneracies, which may serve as indicators of genuine crossings in nearby systems. These findings open new perspectives for the spectral engineering of continuous-variable quantum states.


\section*{Acknowledgements}
The authors thank Prof. Luca Dieci for useful discussions. This work has been partially supported by GNCS--INdAM and by the PRIN2022PNRR n. P2022M7JZW ``SAFER MESH - Sustainable mAnagement oF watEr Resources ModEls and numerical MetHods" research grant, funded by the Italian Ministry of Universities and Research (MUR), by the European Union through Next Generation EU, M4.C2.1.1, CUP H53D23008930001 and by Plan France 2030 through the project ANR-22-PETQ-0013 (OQuLus).

\bibliographystyle{quantum}
\bibliography{biblio_avoided}


\onecolumn
\appendix

\section{On the joint-MVD of a complex symmetric  matrix-valued function depending on a parameter}\label{sec:takagi vs jointMVD}

Here we show that, for a smooth complex symmetric matrix-valued function \begin{align}\label{eq:matfun_appendix}
\omega\in[0,\tau]\mapsto A(\omega)\in\C^\nxn, & \\
A(\omega)=A(\omega)^T \text{ for all } \omega\in[0,\tau], &
\end{align}
with distinct and non-zero singular values, joint-MVD (\cite{DIPU_LAA2024}) 
and smooth Takagi decomposition 
are essentially the same.
Recall that, given $A\in\C^\nxn$, an SVD of $A$ is a decomposition 
\begin{equation}\label{eq:svd_appendix}
    A=U\Sigma V^\dag,
\end{equation}
where $U, V \in\C^\nxn$ are unitary and $\Sigma=\diag(\sigma_1,\sigma_2,\ldots,\sigma_n)$, with non-negative $\sigma_j$'s arranged in non-increasing order. If $A=A^T$, $A$ also admits a Takagi decomposition \begin{equation}\label{eq:takagi_appendix}
    A=W\Sigma W^T,
\end{equation} where $W$ is unitary and $\Sigma$ is as above. Henceforth, we assume that $A(\omega)$ is as in \eqref{eq:matfun_appendix} and has distinct and non-zero singular values for all $\omega$. According to \cite{DIPU_LAA2024}, the factors $U(\omega), \Sigma(\omega), V(\omega)$ of the joint-MVD of $A$ satisfy the following set of differential equations
\begin{equation}\label{eq:Ck_SVD_DAE}
\left\{
\begin{array}{l}
 \dot\Sigma = U^\dag\dot AV-H\Sigma+\Sigma K ,\\
 \dot U = UH, \\
 \dot V = VK, \\
\end{array}.
\right.
\end{equation}
where the matrix functions $H$ and $K$ are skew-Hermitian on $[0,1]$, have off-diagonal entries given by
\begin{equation}\label{eq:offdiagHK_expression}
\left\{\begin{array}{l}
H_{j\ell}=\dfrac{\sigma_\ell(U^\dag\dot A V)_{j\ell}+\sigma_j
(U^\dag\dot A V)_{\ell j}}{\sigma_\ell^2-\sigma_j^2}, \\
K_{j\ell}=\dfrac{\sigma_\ell(U^\dag\dot A V)_{\ell j}+\sigma_j
(U^\dag\dot A V)_{j\ell}}{\sigma_\ell^2-\sigma_j^2},
\end{array}\right.
\end{equation}
for all $j\ne\ell$, and diagonal entries that are purely imaginary and satisfy
\begin{flalign}\label{eq:diagHK}
\qquad\qquad & H_{jj}-K_{jj}=i \dfrac{\Img\big((U^\dag\dot AV)_{jj}\big)}{\sigma_j}, && \\ \label{eq:diagHK_minvar_requ}
& H_{jj}+K_{jj}=0, &&
\end{flalign}
for all $j=1,\dots,n$. We point out that \eqref{eq:Ck_SVD_DAE} must obviously be complemented with appropriate initial conditions (that is, an SVD for $A(0)$). Moreover, we remark that equations~(\ref{eq:offdiagHK_expression}, \ref{eq:diagHK}) must be satisfied by {\bf any} smooth SVD of $A$, whereas equation \eqref{eq:diagHK_minvar_requ} is the extra requirement that ensures that \eqref{eq:svd_appendix} is a joint-MVD.  

According to \cite{DiPaPu_takagi}, $A$ admits also a smooth Takagi decomposition whose smooth factors $W(\omega), \Sigma(\omega)$ satisfy

\begin{equation}\label{eq:takagi_DAE}
\left\{
\begin{array}{l}
\dot \Sigma=W^\dag	\dot A W^* - E\Sigma+\Sigma E^*, \\
\dot W=WE,
\end{array}
\right.
\end{equation}
where $E$ is skew-Hermitian, has off-diagonal entries 
\begin{equation}\label{E_edo}
E_{j\ell}=\frac{\Real(W^\dag\dot A W^*)_{j\ell}}{\sigma_\ell-\sigma_j} + i \frac{\Img(W^\dag\dot A W^*)_{j\ell}}{\sigma_\ell+\sigma_j},
\end{equation} 
for all $j\ne \ell$, and diagonal entries 
\begin{equation}\label{H_edo2}
E_{jj}=i \frac{\Img(W^\dag\dot A W^*)_{jj}}{2 \sigma_j},
\end{equation}
for all $j=1,\ldots,n$.

We now state the main result of this section.
\begin{thm}\label{thm:takagi_is_jMVD} Let $A(\omega)$ be as in \eqref{eq:matfun_appendix}, and suppose it has distinct and non-zero singular values for all $\omega$. Then:
\begin{enumerate}
    \item[i)] any smooth Takagi decomposition of $A$ is a joint-MVD of $A$;
    \item[ii)] any joint-MVD of $A$ with unitary factors $U, V$ satisfying $V(0)=U^*(0)$ is a smooth Takagi decomposition of $A$.
\end{enumerate}
\end{thm}
\begin{proof}
First, we observe that any Takagi decomposition is obviously also an SVD. Let $A(\omega)=W(\omega)\Sigma(\omega)W(\omega)^T$ be a smooth Takagi decomposition, and set $U=W, V=W^*$. Then, we have $A(\omega)=U(\omega)\Sigma(\omega)V(\omega)^\dag$, with $K$ and $H$ in \eqref{eq:Ck_SVD_DAE} related by $K=H^*$. Being skew-symmetric, $H$ and $K$ have purely imaginary diagonal entries; therefore we must have $H_{jj}=-K_{jj}$ for all $j=1,\ldots,n$. This shows (see \eqref{eq:diagHK_minvar_requ}) that the Takagi decomposition of $A$ is in fact a joint-MVD, and concludes the proof of part \emph{i)}.

On the other hand, let $A(\omega)=U(\omega)\Sigma(\omega)V(\omega)^\dag$ be a joint-MVD with $V(0)=U(0)^*$ (or, equivalently, $V(0)^*=U(0)$). Since $A(\omega)=A(\omega)^T$, it follows from \eqref{eq:joint_MVD_integral} that also $A(\omega)=V(\omega)^*\Sigma(\omega)U(\omega)^T$ is a joint-MVD. It follows that the corresponding factors of these two joint-MVDs satisfy the same set of differential-algebraic equations \eqref{eq:Ck_SVD_DAE}, \eqref{eq:offdiagHK_expression}, \eqref{eq:diagHK}, \eqref{eq:diagHK_minvar_requ}, with the same initial conditions. Therefore, we must have, in particular, $V(\omega)^*=U(\omega)$ for all $\omega$. This means that  $A(\omega)=U(\omega)\Sigma(\omega) V(\omega)^\dag=U(\omega)\Sigma(\omega) U(\omega)^T$ is a smooth Takagi decomposition. This concludes the proof.
\end{proof}

\section{Additional figures on rapid variation of the singular vectors}\label{Appendix:singular vectors}

Here appendix we show the behavior of the morphing supermodes of the case studied in subsection~\ref{subsec:four}. In particular, Figure~\ref{fig: n=4 codim=3 eigenvectors} consists of eight panels, each corresponding to a fixed value of $m$, and showing the absolute value of the coefficient $U_{m,j}(\omega)$ across all supermodes $j=\{0,\ldots,7\}$. The vertical black dash-dotted line marks the position of the first avoided crossing at $\omega=0.164\gamma$ occurring between the anti-squeezed supermodes $j=2$ and $j=3$. For the symmetry property of the spectrum of the singular values, the same avoided crossing also occurs between the squeezed supermodes $j=6$ and $j=7$. The vertical black dashed line marks the position of the second avoided crossing at $\omega=0.794\gamma$ occurring between the anti-squeezed supermodes $j=1$ and $j=2$. The same avoided crossing also occurs between the squeezed supermodes $j=5$ and $j=6$. Focusing on the top-left panel of the figure, which displays the absolute value of the $m=0$ coefficient for the eight supermodes, we observe a rapid rotation of the anti-squeezed supermode pair $j=2$ (green line) and $j=3$ (red line) at the first avoided crossing, as well as a rotation of the squeezed pair $j=6$ (pink line) and $j=7$ (gray line).
At the second avoided crossing, a similar rotation occurs between the anti-squeezed supermodes $j=1$ (orange line) and $j=2$ (green line), and between the squeezed modes $j=5$ (brown line) and $j=6$ (pink line).
The same behavior is observed in the panels corresponding to the other coefficients.
\newcommand{\imgwidth}{\textwidth}
\newcommand{\imgheight}{0.7\textwidth}
\begin{figure*}[ht]
    \centering
    \begin{minipage}[b]{0.4\textwidth}
        \centering
        \includegraphics[width=\imgwidth, height=\imgheight]{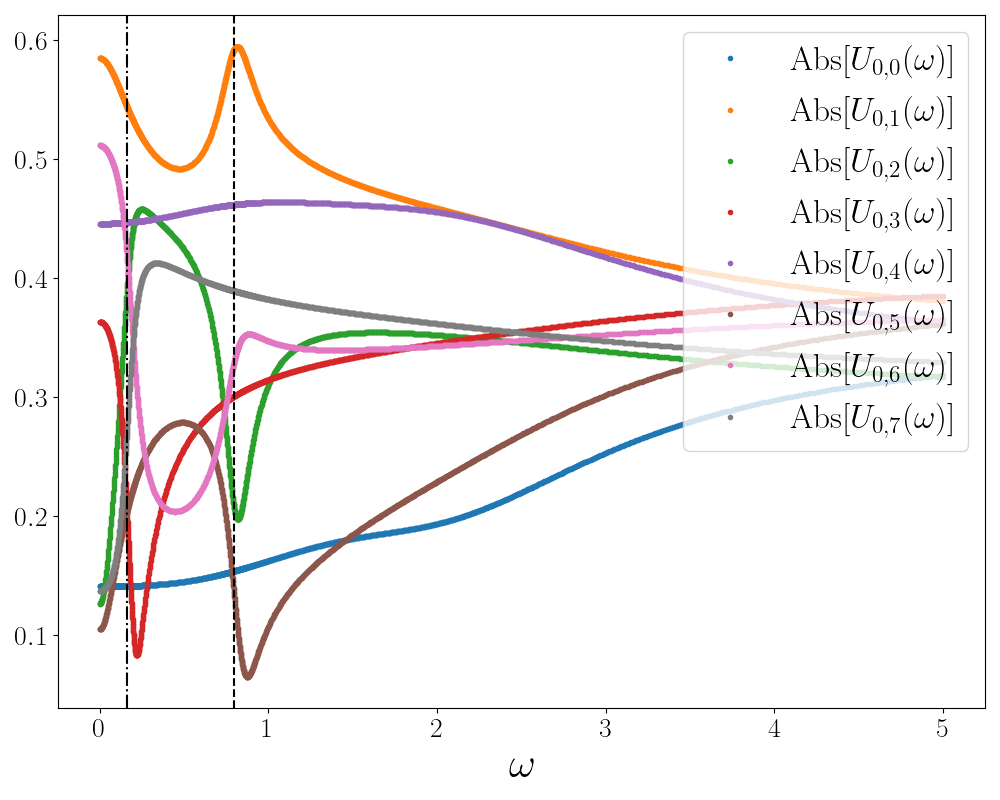}
    \end{minipage}
    \begin{minipage}[b]{0.4\textwidth}
        \centering
        \includegraphics[width=\imgwidth, height=\imgheight]{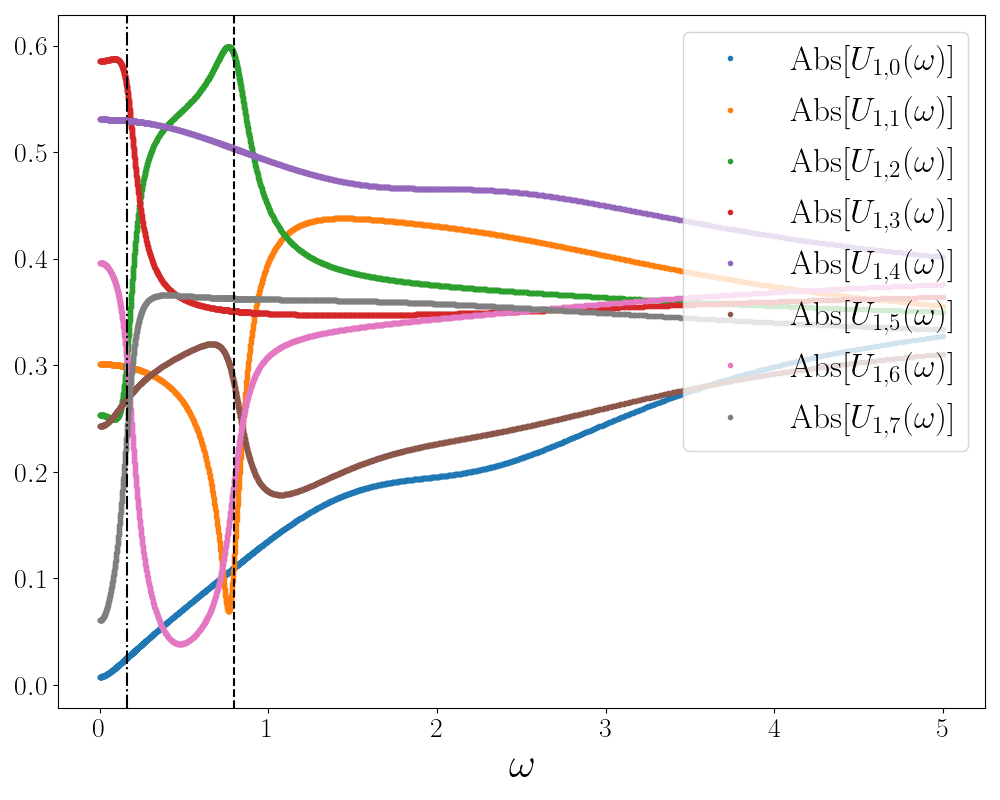}
    \end{minipage}
    
    \begin{minipage}[b]{0.4\textwidth}
        \centering
        \includegraphics[width=\imgwidth, height=\imgheight]{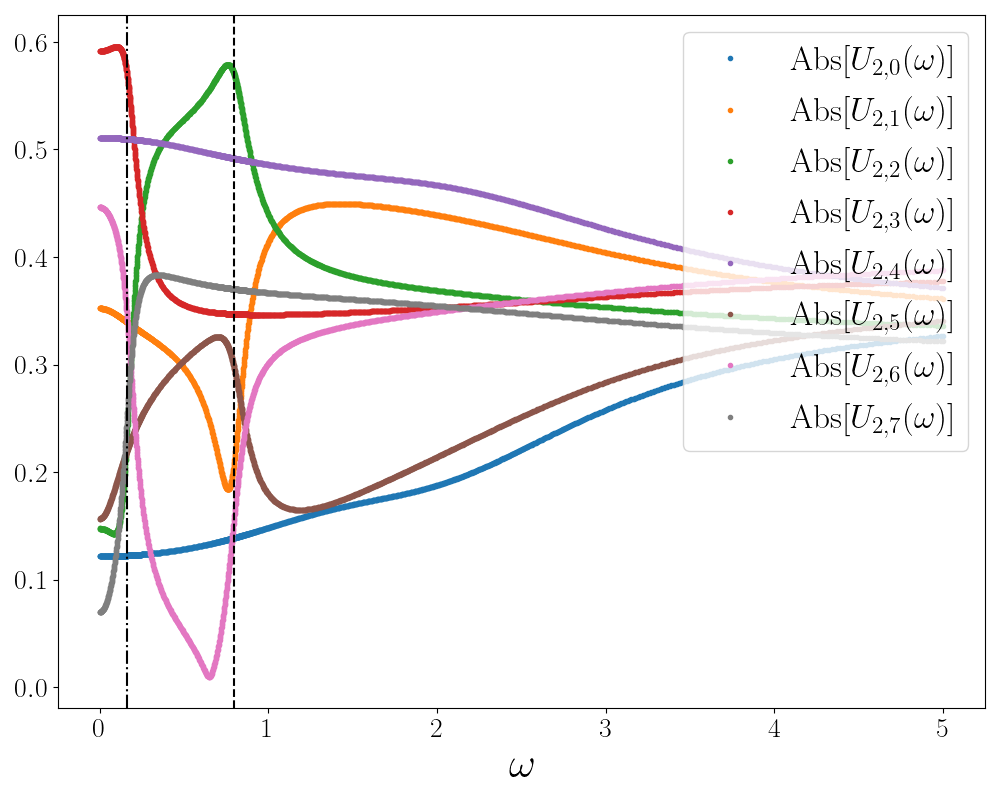}
    \end{minipage}
    \begin{minipage}[b]{0.4\textwidth}
        \centering
        \includegraphics[width=\imgwidth, height=\imgheight]{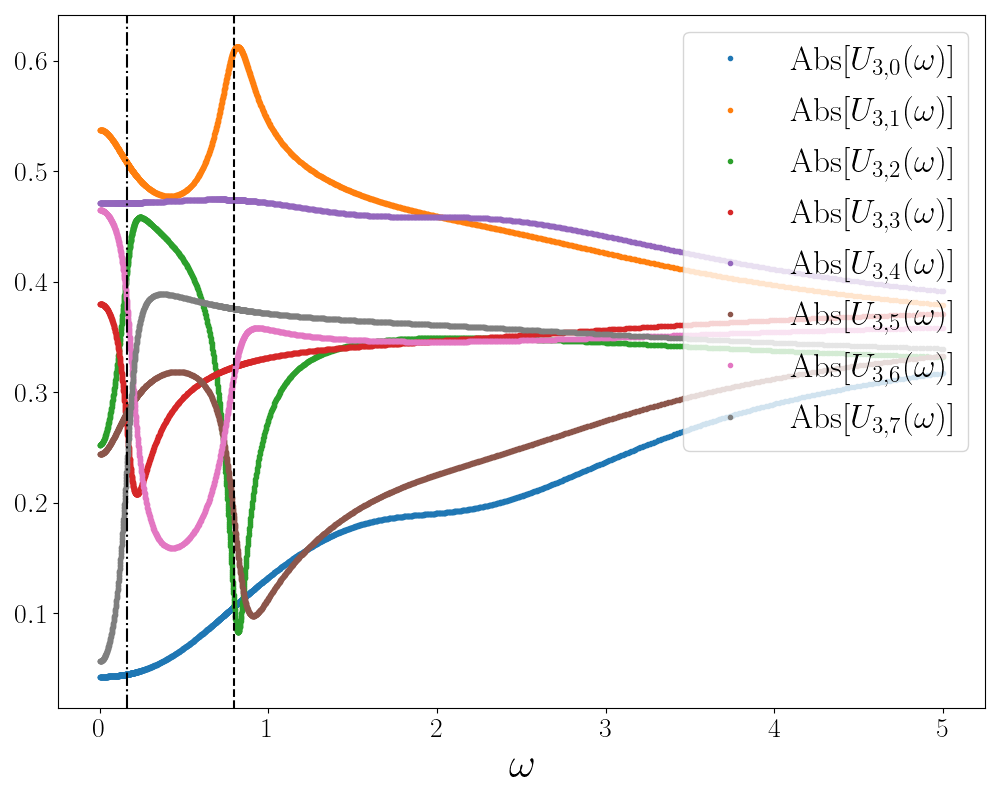}
    \end{minipage}
    
    \begin{minipage}[b]{0.4\textwidth}
        \centering
        \includegraphics[width=\imgwidth, height=\imgheight]{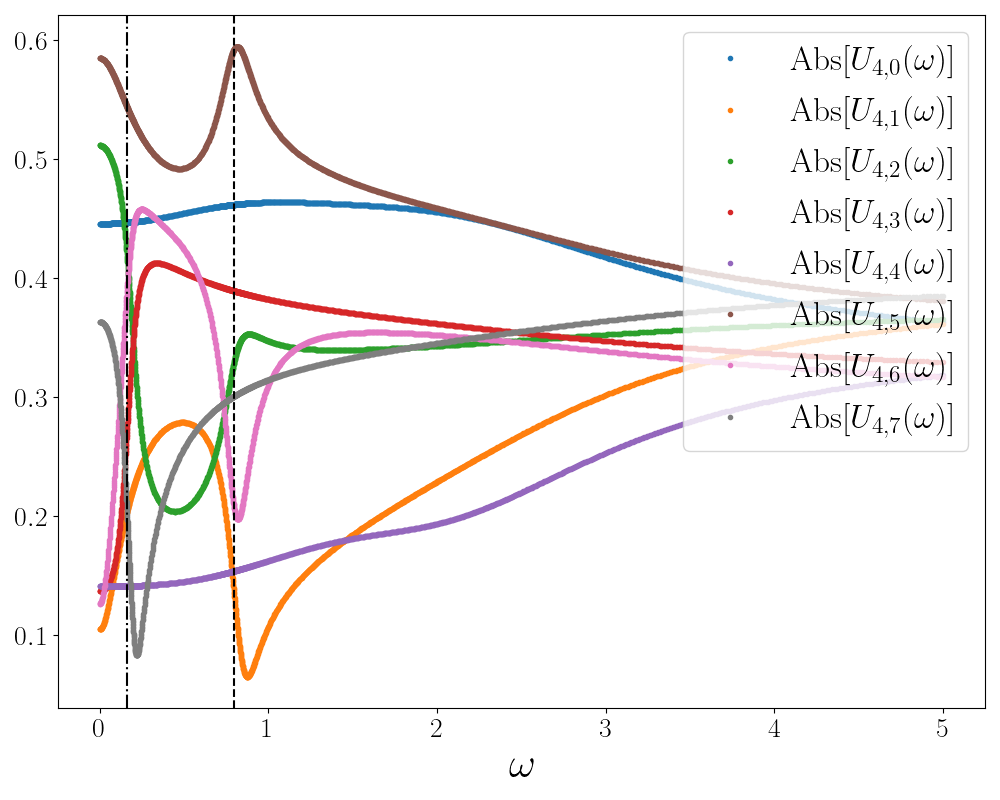}
    \end{minipage}
    \begin{minipage}[b]{0.4\textwidth}
        \centering
        \includegraphics[width=\imgwidth, height=\imgheight]{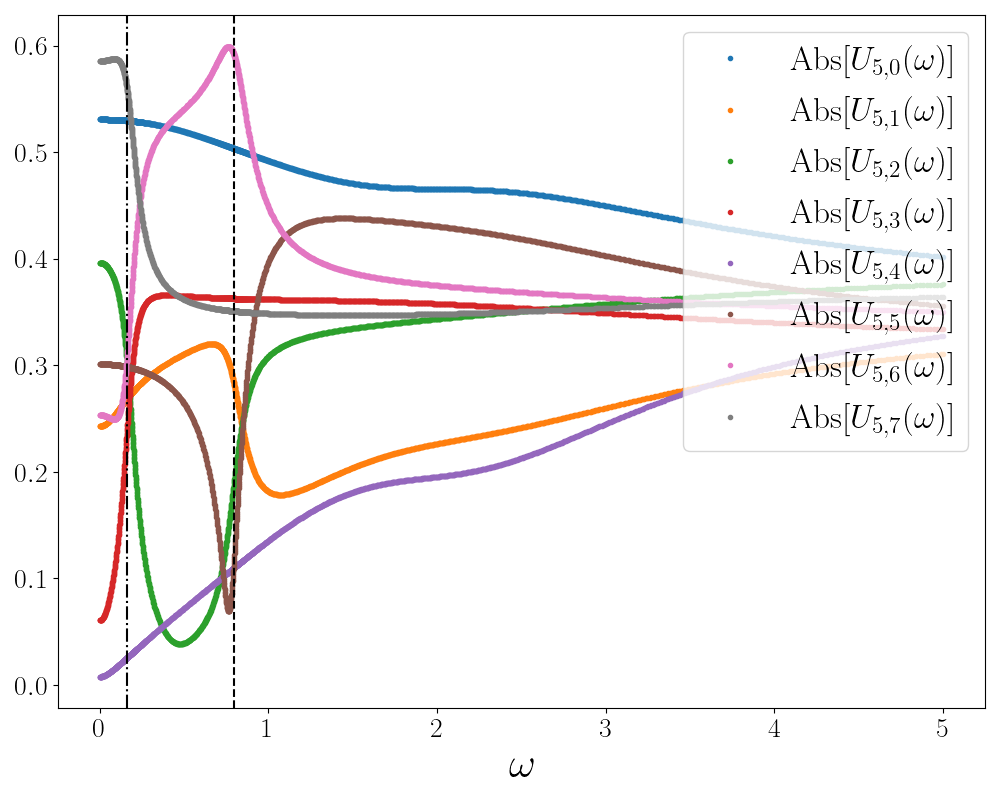}
    \end{minipage}
    
    \begin{minipage}[b]{0.4\textwidth}
        \centering
        \includegraphics[width=\imgwidth, height=\imgheight]{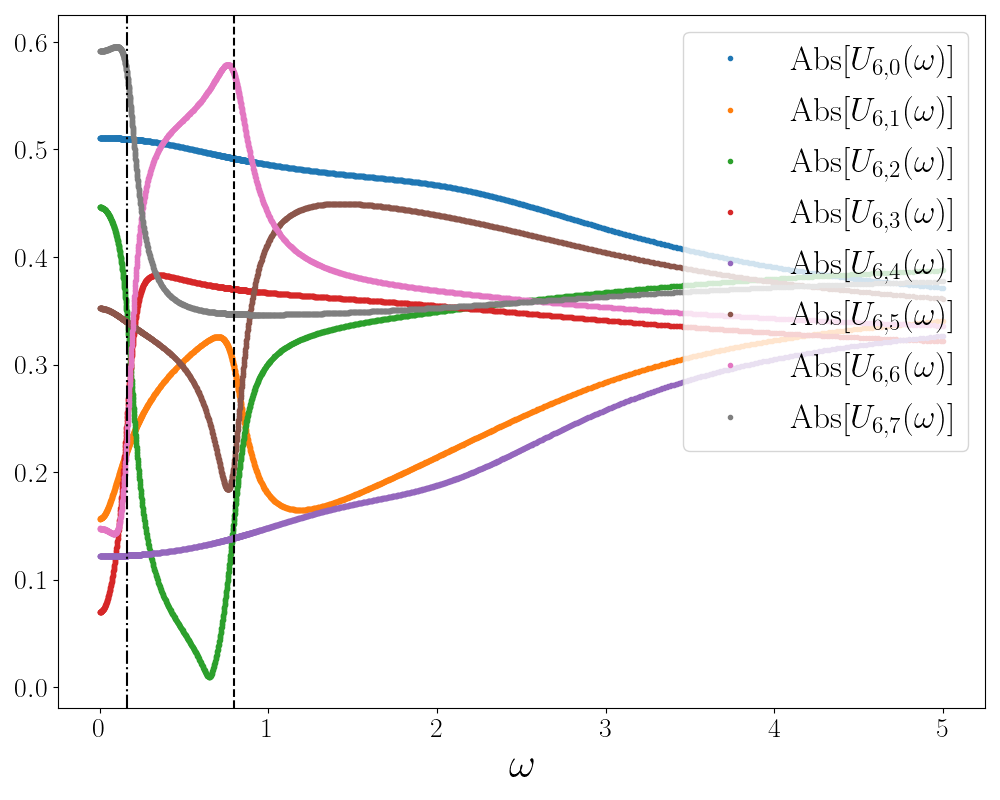}
    \end{minipage}
    \begin{minipage}[b]{0.4\textwidth}
        \centering
        \includegraphics[width=\imgwidth, height=\imgheight]{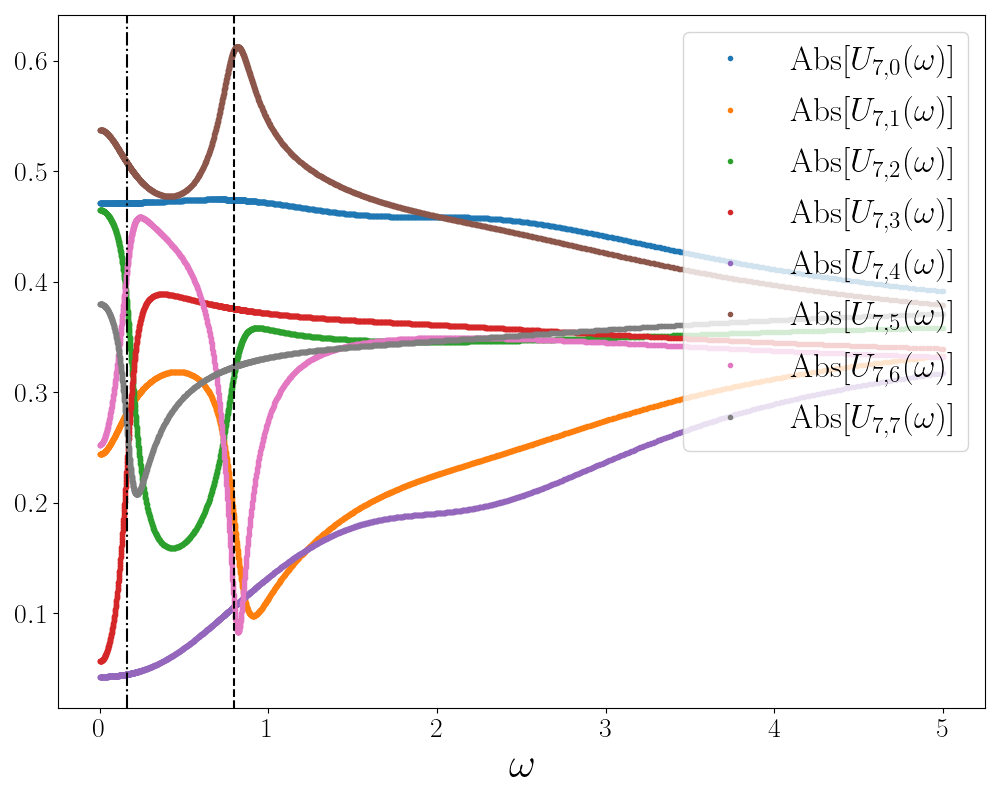}
    \end{minipage}
    
    \caption{Case $n=4$ and $\text{codim}=3$. Absolute value of the components ($m=0,\ldots,7$) of the eight ($n=0,\ldots,7$) eigenvectors $U_{m,n}(\omega)$. The vertical black-dash-dotted line marks the value of $\omega=0.164\gamma$ where avoided crossing between supermodes $j=2$ and $j=3$ occurs. The vertical black-dashed line marks the value of $\omega=0.794\gamma$ where avoided crossing between supermodes $j=1$ and $j=2$ occurs.}
    \label{fig: n=4 codim=3 eigenvectors}
\end{figure*}

\FloatBarrier
\end{document}